\documentclass[12pt]{article}
\pdfoutput=1

\DeclareFontFamily{T1}{calligra}{}
\DeclareFontShape{T1}{calligra}{m}{n}{<->s*[1.44]callig15}{}
\DeclareMathAlphabet\mathcalligra   {T1}{calligra} {m} {n}
\DeclareMathAlphabet\mathzapf       {T1}{pzc} {mb} {it}
\DeclareMathAlphabet\mathchorus     {T1}{qzc} {m} {n}
\DeclareMathAlphabet\mathrsfso      {U}{rsfso}{m}{n}
\DeclareMathAlphabet\mathfrcal      {T1}{frcursive}{m}{it}
\DeclareFontFamily{T1}{frcursive}{}
\DeclareFontShape{T1}{frcursive}{m}{n}{<->s*[1.44]callig15}{}
\DeclareMathAlphabet\mathfrcal      {T1}{frcursive}{m}{it}

\usepackage{amsmath}
\usepackage{amssymb}
\usepackage{amsthm}
\usepackage{graphicx}
\usepackage[dvipsnames]{xcolor}
\usepackage[english]{babel}
\numberwithin{equation}{section}
\usepackage{array}
\usepackage{mathtools}
\usepackage{dsfont}
\usepackage{mathrsfs}
\usepackage{tikz}\usetikzlibrary{matrix,fit}
\usepackage{varwidth}
\usepackage{enumerate}
\usepackage{appendix}
\usepackage{ytableau}

\usepackage{caption}

\usepackage[compat=1.1.0]{tikz-feynman}

\usepackage[margin=1in]{geometry}
\usepackage{nicematrix}

\usepackage[
    backend=bibtex,
    style=alphabetic,
maxbibnames=99,
giveninits=true,
minalphanames=1,
maxalphanames=3,url=false
  ]{biblatex}

  \bibliography{Geodesics}
  
\usepackage{mathabx}
\usepackage{empheq}

\usepackage{setspace}

\usepackage{fontawesome5}

\usepackage{slashed}
\usepackage{upgreek}
\usepackage{appendix}

\usepackage{tabstackengine}

\usepackage{wrapfig}
\usepackage[abs]{overpic}

\usepackage{float}

\usepackage{mathtools}

\fixTABwidth{T}

\newcommand{\dd}{\partial}

\newcommand{\CP}{\mathbb{CP}}
\newcommand{\CC}{\mathbb{C}}

\newcommand{\bea}{\begin{equation}}
\newcommand{\eea}{\end{equation}}
\newcommand{\bear}{\begin{eqnarray}}
\newcommand{\eear}{\end{eqnarray}}
\newcommand{\bearr}{\begin{eqnarray*}}
\newcommand{\eearr}{\end{eqnarray*}}

\newtheorem{prop}{Proposition}
\newtheorem{lem}{Lemma}

\newdimen\mytextwidth
\newcommand\rem[2][cyan!40!green]{\noindent\nobreak\hfil\penalty1000\hfilneg
\mytextwidth=\linewidth\advance\mytextwidth by 2mm
\begin{tikzpicture}[baseline=-\the\dimexpr\fontdimen22\textfont2\relax]\node[outer sep=0pt,draw=black,fill=#1,fill opacity=1,text opacity=1,rectangle,rounded corners]{\begin{varwidth}{\mytextwidth}\textcolor{white}{#2}\end{varwidth}};
\end{tikzpicture}\allowbreak
}

\newcommand\whiterem[2][white!]{\noindent\nobreak\hfil\penalty1000\hfilneg
\mytextwidth=\linewidth\advance\mytextwidth by 2mm
\begin{tikzpicture}[baseline=-\the\dimexpr\fontdimen22\textfont2\relax]\node[outer sep=0pt,draw=black,fill=#1,fill opacity=1,text opacity=1,rectangle,rounded corners,line width=1.5pt]{\begin{varwidth}{\mytextwidth}\textcolor{black}{#2}\end{varwidth}};
\end{tikzpicture}\allowbreak
}

\usepackage{accents}

\newcommand{\Id}{\mathrm{Id}}
\newcommand{\SU}{\mathrm{SU}}

\renewbibmacro{in:}{}

\usepackage{mdframed}

\newbibmacro{string+doi}[1]{
  \iffieldundef{doi}{#1}{\href{http://dx.doi.org/\thefield{doi}}{#1}}}
\DeclareFieldFormat{title}{\usebibmacro{string+doi}{\mkbibemph{#1}}}
\DeclareFieldFormat[article]{title}{\usebibmacro{string+doi}{\mkbibquote{#1}}}

\newmdenv[
  topline=false,
  bottomline=false,
  rightline=false,
  linewidth=2pt,
  skipabove=\topsep,
  skipbelow=\topsep
]{siderules}

\newmdenv[
  topline=false,
  bottomline=false,
  linewidth=2pt,
  skipabove=\topsep,
  skipbelow=\topsep
]{siderulesright}

\makeatletter
\renewcommand{\@seccntformat}[1]{\csname the#1\endcsname.\quad}
\makeatother

\usepackage{setspace}
\usepackage{xpatch}

\makeatletter
\renewcommand{\@chap@pppage}{
  \clear@ppage
  \thispagestyle{plain}
  \if@twocolumn\onecolumn\@tempswatrue\else\@tempswafalse\fi
  \null\vfil
  \markboth{}{}
  {\centering
   \interlinepenalty \@M
   \normalfont
   \MakeUppercase \appendixpagename\par}
  \if@dotoc@pp
    \addappheadtotoc
  \fi
  \vfil\newpage
  \if@twoside
    \if@openright
      \null
      \thispagestyle{empty}
      \newpage
    \fi
  \fi
  \if@tempswa
    \twocolumn
  \fi
}
\makeatother

\definecolor{navycol}{RGB}{100,150,160}
   \definecolor{pinkcol}{RGB}{242,55,55}
   \definecolor{greencol}{RGB}{50,205,50}

   \definecolor{bluecol}{RGB}{30,144,255}

\usepackage{titlesec}

\titleformat*{\section}{\large\bfseries}
\titleformat*{\subsection}{\normalsize\bfseries}
\titleformat*{\subsubsection}{\normalsize\bfseries}
\titleformat*{\paragraph}{\large\bfseries}
\titleformat*{\subparagraph}{\large\bfseries}
\titlespacing{\author}{-5pt}{-5pt}{-5pt}[-5pt]

\makeatletter
\renewcommand\subsubsection{\@startsection{subsubsection}{3}{\z@}
                                     {-3.25ex\@plus -1ex \@minus -.2ex}
                                     {-1.5ex \@plus -.2ex}
                                     {\normalfont\normalsize\bfseries}}
\renewcommand\subsection{\@startsection{subsection}{3}{\z@}
                                     {-3.25ex\@plus -1ex \@minus -.2ex}
                                     {-1.5ex \@plus -.2ex}
                                     {\normalfont\normalsize\bfseries}}                                     
\makeatother

\usepackage{soul}

\DeclareRobustCommand{\db}[1]
{
	{\begingroup
		\sethlcolor{Yellow}\hl{DB: #1}
		\endgroup}
}

\DeclareFontFamily{U}{solomos}{}
\DeclareErrorFont{U}{solomos}{m}{n}{10}
\DeclareFontShape{U}{solomos}{m}{n}{
  <-> s*[1.1]  gsolomos8r
}{}

   \usepackage{stmaryrd}

\usepackage{tikz}
\usetikzlibrary{arrows.meta}

\usepackage{indentfirst}

\usepackage{tocloft}
\let \savenumberline \numberline
\def \numberline#1{\savenumberline{#1.}}

\usepackage{etoolbox}
\patchcmd{\tableofcontents}{\@starttoc}{\vspace{-0.3cm}\@starttoc}{}{}

\usepackage{xurl}
\usepackage{scrextend}

\usepackage[hyperfootnotes=]{hyperref}
\hypersetup{
    colorlinks=true,
    linkcolor=blue,
    filecolor=magenta,      
    urlcolor=cyan,
    breaklinks=true
    }

\begin{document}

\title{\vspace{-1.0cm} The classical and quantum particle \\ on a flag manifold\\ \vspace{1.5cm} }

\author{Dmitri Bykov$^{\,a, \,b,\,}$\footnote{Emails:
 bykov@mi-ras.ru, dmitri.v.bykov@gmail.com} \qquad\qquad Andrew Kuzovchikov$^{\,c\,}$\footnote{Email:
 andrkuzovchikov@mail.ru}
\\  \vspace{0cm}  \\
{\small $a)$ \emph{Steklov
Mathematical Institute of Russian Academy of Sciences,}} \\{\small \emph{Gubkina str. 8, 119991 Moscow, Russia} }\\
{\small $b)$ \emph{Institute for Theoretical and Mathematical Physics,}} \\{\small \emph{Lomonosov Moscow State University, 119991 Moscow, Russia}}\\
{\small $c)$ \emph{Saint Petersburg State University,}} \\{\small \emph{Universitetskaya nab. 7/9, 199034 St.Petersburg, Russia}}
}

\date{}

{\let\newpage\relax\maketitle}

\maketitle


    \ytableausetup{centertableaux}
    
    \textbf{Abstract.} In the present paper we consider two related problems, i.e. the description of geodesics and the calculation of the spectrum of the Laplace-Beltrami operator on a flag manifold. We show that there exists a family of invariant metrics such that both problems can be solved simply and explicitly. In order to determine the spectrum of the Laplace-Beltrami operator, we construct natural, finite-dimensional approximations (of spin chain type) to the Hilbert space of functions on a flag manifold.

    \newpage
\tableofcontents
\addtocontents{toc}{\vspace{0.5cm}}
    
    \section{Introduction}
    Searching for geodesics on Riemannian manifolds is an interesting, but rather difficult problem due to nonlinearity of the geodesic equations. This problem is particularly interesting when the manifold in question is a homogeneous space. There is vast literature devoted to the study of  integrability of geodesic flows on some of them, cf. the well-known work \cite{Thimm, Paternain,AVBolsinov_1998} as well as~\cite[Chapter 3]{Cheeger} for general information about geodesics on homogeneous spaces.
    
    There are also articles on non-commutative integrability, e.g. \cite{Bolsinov}: here the main ideas and the connection with ordinary integrability are described in \cite{BolsinovTop}. Other questions related to geodesics have also been studied, for example, the case when the manifold is a space of geodesic orbits~\cite{Jovanovi__2011}. In addition, topological obstructions to integrability are considered in~\cite{BolsinovTop}. 

    In this paper we are interested in a special class of homogeneous spaces, namely the flag manifolds. Let us recall the basic notions. A flag is a sequence of nested linear subspaces in $\mathbb{C}^N$,
    \begin{gather}
        0 \subset L_1 \subset L_2 \subset \ldots \subset L_m = \mathbb{C}^N\,,
    \end{gather}
    of given dimensions $\mathrm{dim} L_a = d_a$. The flag manifold is the manifold of such nested subspaces
    \begin{gather}
        \mathcal{F}_{d_1,d_2,\ldots,d_m}=\{0 \subset L_1 \subset \ldots \subset L_m = \mathbb{C}^N\}.
    \end{gather}
    This is a homogeneous space which can be presented as a quotient space of the unitary group \cite{Affleck_2022}: 
    \begin{gather}
        \mathcal{F}_{d_1,d_2,\ldots,d_m}=
        \frac{\mathrm{U}(N)}{\mathrm{U}(n_1)\times \mathrm{U}(n_2)\times \ldots \times \mathrm{U}(n_m)},
    \end{gather}
    where $n_i=d_i-d_{i-1}$, $i=1,2,\ldots,m$ (we set $d_0=0$). Basic information about flags\footnote{For brevity sometimes we refer to flag manifolds simply as flags, with the hope that this will not confuse the reader.} can be found in \cite{ArvGeometryOfFlags}.
    
    We say that $\mathcal{F}_{1,2,\ldots,N}=\frac{\mathrm{U}(N)}{\mathrm{U}(1)^{N}}$ is a manifold of complete flags. A point on it can be parameterized by an ordered set of $N$ pairwise orthogonal vectors $\{u_i\in \mathbb{CP}^{N-1}\}_{i=1}^N$. For convenience, we can assume that the vectors are normalised, i.e. $\Bar{u}_a \circ u_b := \sum_{j=1}^N \Bar{u}^j_a u^j_b = \delta_{ab}$. Furthermore, since the vectors take values in the projective space, they are determined up to a multiplication by a phase factor. For generic values of the $d_i$'s  $\mathcal{F}_{d_1,d_2,\ldots,d_m}$ are  manifolds of partial flags.

    Geodesics on flag manifolds have also been previously studied in the literature. In~\cite{Paternain} it was shown that an arbitrary invariant metric on the flag manifold $\mathcal{F}_{1,2,3}=\frac{\mathrm{U}(3)
    }{\mathrm{U}(1)^{\times 3}}$ admits a fully integrable geodesic flow if not every geodesic is an orbit of a one-parameter subgroup of $\SU(3)$. In \cite{AlekseevskyGO} the authors studied those  metrics, for which flag manifolds are spaces of geodesic orbits. Finally, the article \cite{CohenEq} focuses on the study of equigeodesics on flag manifolds.

    In this paper we consider special metrics on manifolds of complete (or partial) flags, such that the geodesic equations can be solved explicitly. For these metrics it is shown that the eigenvalues of the Laplace-Beltrami operator can also be found explicitly. We note that the so-called normal metric on a flag manifold  is a point in the space of metrics that we study. In this latter case the eigenvalue problem had been previously solved in \cite{Yamaguchi}.

    In order to study geodesics and eigenvalues of the Laplace-Beltrami operator on $\mathcal{F}_{1,2,\ldots,N}$ we use the following idea stated in \cite{BykLagEmb}. One should consider an $\SU(N)$-spin chain, and in a certain  `infinite spin' limit its Hamiltonian transforms into the Hamiltonian of a free particle on the complete flag manifold. The solution of the classical problem provides the geodesics, while the corresponding quantum problem allows computing the spectrum of the Laplace-Beltrami operator. 

    The paper is organised as follows. In Section~\ref{SU(2)Case} we describe the geodesics and the spectrum of the Laplace-Beltrami operator on a sphere in a formalism that can be generalized to the $\SU(N)$ case and is used later on. Before turning to the most general case, the $\mathrm{SU}(3)$ case is studied in detail in Section~\ref{SU3Case}, i.e. all geodesics on the corresponding complete flag manifold are explicitly determined and the spectrum of the Laplace-Beltrami operator is calculated. Furthermore, the transition from $\mathcal{F}_{1,2,3}$ to the projective space $\CP^2$ (which, in this context, should be viewed as a partial flag manifold) is investigated. This method admits a natural generalization to the case of higher $N$. The corresponding scheme is described in detail in Section~\ref{Nsec}. In particular, in Sections~\ref{SpinChain}-\ref{flagMetr} we outline a method for computing the spectrum of the Laplace-Beltrami operator. Proposition~\ref{geodsolproof} is devoted to the explicit solution of the geodesic equations by induction in $N$. Finally, in Section~\ref{partialFlags} we study the transition from $\mathcal{F}_{1,2,\ldots,N}$ to an arbitrary partial flag manifold.
    \section{$\SU(2)$ case} \label{SU(2)Case}
    Let us recall the basic properties of the well-known behaviour of a free particle (classical or quantum) on the sphere $\mathcal{S}^2$, which is the simplest complete flag manifold~$\mathcal{F}_{1,2}$. The ideas presented here are further developed in what follows and applied to~$\mathcal{F}_{1,2,\ldots,N}$.

    \subsection{Spectrum of the quantum Hamiltonian.}  
    The case of a quantum particle on $\mathcal{S}^2$ can be obtained as a large spin limit of a simple system consisting of two $\SU(2)$ spins. Its Hamiltonian is\footnote{In the following we assume summation w.r.t. repeated indices.}
    \begin{gather}\label{su2Ham}
        \mathcal{H} ={1\over 2}\left(S^a_1+S^a_2\right)^2 =S^a_1 S^a_2+\mathrm{const}\,,
    \end{gather}
    where $S^a_i$ ($a=1,2,3$) are generators of the $\mathfrak{su}(2)$ algebra in the representation\footnote{$S^a_i S^a_i$ is a quadratic Casimir operator acting in the irreducible representation $T^{p\over 2}$ and therefore proportional to the identity operator.} $T^{\frac{p}{2}}$ of spin $\frac{p}{2}$.
    Thus, the operator $\mathcal{H}$ acts in the tensor product of representations $T^{\frac{p}{2}} \otimes T^{\frac{p}{2}}$.

    It is clear that the Hamiltonian~(\ref{su2Ham}) is proportional to the quadratic Casimir operator $C_2=S^a S^a$, where $S^a= S^a_1+S^a_2$ are generators of the $\mathfrak{su}(2)$ algebra in the representation considered. In the following $\mathrm{SU}(N)$ generalization  we will also consider Hamiltonians associated with quadratic Casimir operators.
    
    A well-known result of representation theory is
    \begin{gather}\label{S2pinf}
        \underset{p\to\infty}{\mathrm{lim}}\,T^{\frac{p}{2}} \otimes T^{\frac{p}{2}}=\underset{p\to\infty}{\mathrm{lim}}\,\bigoplus_{j=0}^p T^{j} = L^2(\mathcal{S}^2),
    \end{gather}
    where $L^2(\mathcal{S}^2)$ is the space of square integrable functions on the sphere. Details of the above facts from representation theory can be found, for example, in \cite{Isaev:2018xcg}.

    Thus, the considered system describes a quantum particle on the sphere $\mathcal{S}^2 \cong \mathbb{CP}^1$ in the limit $p\to\infty$.
    The spectrum of $\mathcal{H}$ is proportional to the spectrum of the Casimir operator acting in the space~(\ref{S2pinf}) and is given by a well-known formula 
    \begin{gather}
        \Lambda_j = j(j+1),\quad j\in \mathbb{N}_0.
    \end{gather}
    Clearly, this is also the spectrum of the Laplace-Beltrami operator on the sphere $\mathcal{S}^2$. It is natural to assume that operator~(\ref{su2Ham}) has the Laplace-Beltrami operator as its limit when $p\to\infty$. We restrict ourselves to this observation for the moment: a proof of this claim in the general ($\mathrm{SU}(3)$ and $\mathrm{SU}(N)$) case will be provided  in Section \ref{MetricsFlag}.

    \subsection{Geodesics.} 
    In order to describe the motion of a classical particle on the sphere, we will make use of an embedding
    \begin{gather}
    \mathbb{CP}^1 \hookrightarrow \mathbb{CP}^1 \times \mathbb{CP}^1.
    \end{gather}
    Let us assume that each of the spheres on the right-hand side is parameterized by a unit vector $u_1, u_2 \in \CC^2$, defined up to a phase multiplier. Then the embedding can be described as the orthogonality condition for these two vectors:
    \begin{gather}
    \Bar{u}_1\circ u_2=0.
    \end{gather}
    In the following it is convenient to assume that these vectors are normalised, i.e. $\Bar{u}_i \circ u_j = \delta_{ij}$. In other words, the matrix
    \begin{gather}\label{U2U12quot}
        (\;u_1\;u_2\;)\in {\mathrm{U}(2)\over \mathrm{U}(1)\times \mathrm{U}(1)}\,\cong\,\CP^1.
    \end{gather}
    In these terms the $\SU(2)$-invariant metric on the sphere can be written as
    \begin{gather}\label{S2metr}
        d s^2=
        \left|\Bar{u}_1\circ du_2\right|^2.
    \end{gather}
    Since the $\SU(2)$-invariant metric on $\mathbb{CP}^1$ is unique (up to a constant multiplier), (\ref{S2metr}) is just an unusual form of the standard metric on the sphere.
    This formalism naturally generalizes to more general cases of flag manifolds, which is why we study the two-dimensional sphere in this less familiar formalism.
    
    The geodesic equations reduce to
    \begin{gather}\label{S2EoM}
\frac{D}{Dt}\left(\Bar{u}_2\circ \dot{u}_1\right)=\frac{d}{dt}\left(\Bar{u}_2\circ \dot{u}_1\right)-\Bar{u}_2\circ \dot{u}_1\left(\Bar{u}_1\circ \Dot{u}_1-\Bar{u}_2\circ \Dot{u}_2\right)=0\,,
\end{gather}
where $\frac{D}{Dt}$ is a covariant derivative w.r.t. the gauge group $U(1)^2$, which appears in the denominator of~(\ref{U2U12quot}).
    In the following we use the gauge fixing condition $\Bar{u}_1\circ \Dot{u}_1=\Bar{u}_2\circ \Dot{u}_2=0$, so that $\frac{D}{Dt}=\frac{d}{dt}$. Solving (\ref{S2EoM}), we get $\Bar{u}_2\circ \dot{u}_1 = \mathrm{const}$. Since the vectors $u_1$ and $u_2$ form a basis, we can write
    \begin{gather}
        \dot{u}_1=u_1 (\Bar{u}_1\circ \Dot{u}_1) + u_2 (\Bar{u}_2\circ \Dot{u}_1) = u_2 (\Bar{u}_2\circ \Dot{u}_1), \nonumber \\
        \dot{u}_2=u_1 (\Bar{u}_1\circ \Dot{u}_2) + u_2 (\Bar{u}_2\circ \Dot{u}_2) = u_1 (\Bar{u}_1\circ \Dot{u}_2).
    \end{gather}
    Evolution of these vectors can therefore be easily found:
    \begin{gather}
        \begin{pmatrix}
           u_1(t)\\
           u_2(t)
        \end{pmatrix} = 
        \exp \left[
        \begin{pmatrix}
            0 & A^0_{21}\\
            -\Bar{A}^0_{21} & 0
        \end{pmatrix}t
        \right]
        \begin{pmatrix}
           u_1(0)\\
           u_2(0)
        \end{pmatrix}\,,
        \label{ClassicalSolutionS2}
    \end{gather}
    where $\Bar{u}_2 \circ \Dot{u}_1 = A^0_{21}$.
    \section{$\SU(3)$ case}\label{SU3Case}
 
    This section generalizes the main ideas from the previous section to the $\SU(3)$ case and to the manifold $\mathcal{F}_{1,2,3}$ accordingly. To this end we analyze a specific $\SU(3)$ spin chain and exhibit its relation to the free particle problem on $\mathcal{F}_{1,2,3}$. 
    \subsection{$\SU(3)$ spin chain consisting of three spins.}\label{definitionOfTheSU(3)SpinChain}
    Consider a system of three $\SU(3)$ spins with Hamiltonian of the form
    \begin{gather}
        \mathcal{H}
        =
        \alpha S^a_1 S^a_2 + \beta S^a_2 S^a_3 + \gamma S^a_1 S^a_3 + \mathrm{const},\label{spinChainHamiltonian}  
    \end{gather}
    where $\alpha, \beta, \gamma$ are real, non-negative numbers. For convenience, we introduce an additional additive constant by analogy with equation~(\ref{su2Ham}). Generators $S^a_i$ ($a=1,\ldots,8$) of the $\mathfrak{su}(3)$ algebra are taken in the $p$-th symmetric power of the associated fundamental representation, $\mathrm{Sym}(p)$. This corresponds to a Young diagram with one row of length $p$. Thus, $\mathcal{H}$ acts in the space 
    \\
    \begin{gather}\label{Vpdef}
        \mathrm{V}(p) =
        \underbrace{
        \begin{ytableau}
               ~ & ~ & \dots & ~
        \end{ytableau}
        }_{p}^{1}
        \otimes
        \underbrace{
        \begin{ytableau}
               ~ & ~ & \dots & ~
        \end{ytableau}
        }_{p}^{2}
        \otimes
        \underbrace{
        \begin{ytableau}
               ~ & ~ & \dots & ~
        \end{ytableau}
        }_{p}^{3} 
        = \mathrm{Sym}(p)^{\otimes 3}.
    \end{gather}
   
    Based on the analogy with the sphere~(\ref{S2pinf}), eventually it becomes necessary to take the limit $p \rightarrow \infty$. Specifically, we will show that
    \[
    \underset{p\to\infty}{\mathrm{lim}}\,\mathrm{V}(p) =L^2(\mathcal{F}_{1,2,3})\,,
    \]
    which is a natural generalization of the statement~(\ref{S2pinf}) for the sphere $\mathcal{S}^2=\mathcal{F}_{1,2}$.
    
    Let us introduce $\mathrm{Irr}\big(\mathrm{Sym}(p)^{\otimes n}\big)$ -- the set of Young diagrams corresponding to irreducible representations in $\mathrm{Sym}(p)^{\otimes n}$. Note that for each diagram in $\mathrm{Irr}\big(\mathrm{V}(p)\big)$ there is a diagram in $\mathrm{Irr}\big(\mathrm{V}(p+1)\big)$ with an additional column of three boxes. They both correspond to the same representation. Thus, 
    $\mathrm{Irr}\big(\mathrm{V}(p)\big)\hookrightarrow \mathrm{Irr}\big(\mathrm{V}(p+1)\big)$. All of this follows from the rules for multiplying Young diagrams and mapping those to irreducible representations, cf.~\cite[Appendix A.1]{fulton1991representation}. 

    \subsection{Derivation of the metric on $\mathcal{F}_{1,2,3}$.}\label{MetricsFlag}
    In this section we find a connection between the considered spin chain and the free particle problem on $\mathcal{F}_{1,2,3}$. In particular, it is shown that the spectrum of the Hamiltonian $\mathcal{H}$ coincides (in the limit $p\to \infty$) with that of the Laplace-Beltrami operator on $\mathcal{F}_{1,2,3}$ with a certain metric.

    The proof is based on the formalism of path integration\footnote{The formalism of path integration was applied to finite-dimensional systems, for example, in~\cite{ALEKSEEV1988391}.}. We are interested in the spin chain partition function $\mathrm{Tr}_{\mathrm{V}(p)}(e^{-\tau\mathcal{H}})$ and we work with generalized coherent states to write out the path integral (see~\cite{Perelomov2002} for general information on coherent states and~\cite{Affleck_2022} for an application to spin chains). In our case, these states are parameterized by the set $(z_1,z_2,z_3) \in \left(\mathbb{CP}^2\right)^{\times 3}$. 
    
    One can write an expression for the partition function as a path integral with an action of the form \cite{BykLagEmb}
    \begin{gather}
        \mathcal{S}=ip\int_0^{\tau} dt\,
        \left(
            \frac{\Bar{z}_1\circ \dot{z}_1}{\left|z_1\right|^2}+
            \frac{\Bar{z}_2\circ \dot{z}_2}{\left|z_2\right|^2}+
            \frac{\Bar{z}_3\circ \dot{z}_3}{\left|z_3\right|^2}
        \right)
        -
        \nonumber \\
        -
        p^2\int_0^{\tau} dt\, 
        \left(
            \alpha \frac{\left|\Bar{z_1}\circ z_2\right|^2}{\left|z_1\right|^2\left|z_2\right|^2}+
            \beta \frac{\left|\Bar{z_2}\circ z_3\right|^2}{\left|z_2\right|^2\left|z_3\right|^2}+
            \gamma \frac{\left|\Bar{z_1}\circ z_3\right|^2}{\left|z_1\right|^2\left|z_3\right|^2}
        \right), \label{InitialAction}
    \end{gather}
    where $z_i \in \mathbb{CP}^2$, and periodic boundary conditions should be imposed: $z_i(t+\tau)=z_i(t)$.

Further consideration is based on the following proposition:
 \begin{prop}[\cite{BykLagEmb}]\label{lagrtheorem}
       $\mathcal{F}_{1,2,3}$ is a Lagrangian submanifold of $(\mathbb{CP}^2 \times \mathbb{CP}^2 \times \mathbb{CP}^2, \Omega)$, where the symplectic form $\Omega$ is a sum of Fubini-Study forms. 
    \end{prop}
Recall the proof:
\begin{proof}
    As before in Section~\ref{SU(2)Case}, assume that a point in $(\CP^2)_i$ is parameterized by the unit vector $u_i\in\CC^3$, determined up to  multiplication by a phase factor. Let us write out the momentum map $\mu: \mathbb{CP}^2 \times\mathbb{CP}^2 \times\mathbb{CP}^2\to \mathfrak{su}(3)$ corresponding to the diagonal action of $\SU(3)$:
    \begin{gather}\label{momapsu3}
        \mu=\sum\limits_{i=1}^3\,u_i\otimes \Bar{u}_i-\Id.
    \end{gather}
    Obviously, the manifold $\mu^{-1}(0)$ is invariant under the action of the group. Moreover, it is well-known that every orbit of the group action lying in $\mu^{-1}(0)$ is isotropic. However,~(\ref{momapsu3}) is a `partition of unity', so that in our case $\mu^{-1}(0)$ consists of pairwise orthogonal vectors in $\CP^2\times \CP^2\times \CP^2$, which is nothing but $\mathcal{F}_{1,2,3}$. Since $\mathcal{F}_{1,2,3}$ is a homogeneous space of $\SU(3)$, $\mu^{-1}(0)$ consists of a unique orbit. Therefore $\Omega\big|_{\mu^{-1}(0)}=0$. Due to the fact that $\mathrm{dim} \,\mathcal{F}_{1,2,3}=6$, it is clear that this isotropic submanifold is  Lagrangian.
\end{proof}

There is even a more general fact at play:

 \begin{prop}\label{subbundle}
       Let $\mathcal{X}\subset \mathbb{CP}^2 \times \mathbb{CP}^2 \times \mathbb{CP}^2$ be the set of all triples of lines in~$\mathbb{C}^3$ passing through the origin and not laying in the same plane (i.e. such that ${\mathrm{Det} (z_1,z_2,z_3)\neq 0}$). Then $\mathcal{X}$ is symplectomorphic to an open subset of $T^\ast \mathcal{F}_{1,2,3}$.
    \end{prop}

    \begin{proof}
        Let $Z:=(z_1,z_2,z_3)$ be a nondegenerate matrix (i.e. $Z\in \mathcal{X}$). Assume the columns are normalised: $|z_1|^2=|z_2|^2=|z_3|^2=1$. In this case the symplectic form on $\mathbb{CP}^2 \times \mathbb{CP}^2 \times \mathbb{CP}^2$ can be written as follows:
        \begin{gather}\label{CP^2symplform}
            \Omega=i\,\sum\limits_{k=1}^3\,d\Bar{z}_k\wedge dz_k=i\,\mathrm{Tr}\left(dZ^{\dagger}\wedge dZ\right).
        \end{gather}
        According to the polar decomposition theorem, $Z$ can be uniquely represented as $Z=UH$, where $U$ is a unitary and $H$ a positive-definite Hermitian matrix. Obviously, here $H=(Z^{\dagger}Z)^{1/2}$ (the square root is well-defined) and $U=Z(Z^{\dagger}Z)^{-1/2}$. Note also that, due to the normalization of the vectors $z_i$, the diagonal elements of $K:=H^2$ are all equal to~$1$. Substituting this decomposition into~(\ref{CP^2symplform}), after some simple manipulations one obtains
        \begin{gather}
            \Omega=i \,d \,\mathrm{Tr}\left(K\,U^{\dagger}dU\right)=i \,\sum\limits_{k=1}^3 d\Bar{u}_k\wedge du_k+i\,d\,\left(\sum\limits_{j\neq k} K_{jk}\, \Bar{u}_j du_k\right),
        \end{gather}
        where $u_i$ are the columns of $U$. The first term in this equation vanishes because $\mathcal{F}_{1,2,3}$ is Lagrangian (Proposition~\ref{lagrtheorem}). The second term is the standard symplectic form on a cotangent bundle. However, note that the matrix elements $K_{jk}$ are not completely arbitrary, since the matrix $K$ is positive-definite by construction. This condition distinguishes the open subset of $T^\ast \mathcal{F}_{1,2,3}$. 
    \end{proof}

    A straightforward generalization of this statement is obviously true for the embedding $\mathcal{F}_{1, \ldots , n}\subset (\CP^{n-1})^{\times n}$. In the case $n=2$ the positive-definiteness condition of the matrix $K$ takes the specially simple form $|K_{12}|<1$ (i.e. the open subset is a disc subbundle in this case).  
\vspace{0.3cm}

Let us use the main ideas of the proof of Proposition \ref{subbundle} to simplify the action~(\ref{InitialAction}). Consider the matrix $Z:=(z_1,z_2,z_3)$ composed of the vectors featuring in~(\ref{InitialAction}) and assume that $Z$ is non-degenerate, i.e. $Z\in \mathcal{X}$. In terms of the path integral this assumption can be justified by the fact that the set $\mathcal{X}$ is dense in $\CP^2\times \CP^2\times \CP^2$. Therefore the value of the integral does not change if we replace $\CP^2\times \CP^2\times \CP^2$ with $\mathcal{X}$ in the domain of integration.

Let $|z_1|^2=|z_2|^2=|z_3|^2=1$, use polar decomposition $Z=UH$ and introduce $K := H^2$ as in the proof. Let $u_i$ be the columns of the unitary matrix $U$, i.e. $\Bar{u}_i \circ u_j=\delta_{ij}$. With these definitions, (\ref{InitialAction}) can now be rewritten in the form
    \begin{gather} \label{remadeAction}
        \mathcal{S}=ip\int_0^{\tau} dt\,\mathrm{Tr}\left(H\dot{H}+KU^{\dagger} \dot{U}\right)-p^2 \int_0^{\tau} dt\, \left(\alpha K_{12}K_{21}+\beta K_{23}K_{32}+\gamma K_{13}K_{31}\right).
    \end{gather}

    The first term is a full derivative and the corresponding integral vanishes due to periodic boundary conditions. As previously noted, the diagonal elements of $K$ are all equal to one. Therefore  one can isolate in (\ref{remadeAction}) the term  $ip\displaystyle\int_0^{\tau} dt\,\left(\Bar{u}_1\circ \Dot{u}_1+\Bar{u}_2\circ \Dot{u}_2+\Bar{u}_3\circ \Dot{u}_3\right)$, whose integrand is  a full derivative as well, since $d \left(\sum\limits_{i=1}^3 \Bar{u}_i du_i\right)=0$. 
    To deal with the remaining terms, one can make the substitution $pK_{ij}\rightarrow K_{ij}, \,i\neq j $ and integrate over~$K_{ij}$. As a result, one arrives at
    \begin{gather}
        \mathcal{S}=
        \int_0^{\tau} dt\, 
        \left(
            \frac{\left|j_{12}\right|^2}{\alpha}+
            \frac{\left|j_{23}\right|^2}{\beta}+
            \frac{\left|j_{13}\right|^2}{\gamma}
        \right),\label{FinalAction}\\
        \textrm{where}\quad j_{ik}:= \Bar{u}_k \circ \dot{u}_i\,.
    \end{gather}
    A reservation is in order, though. When integrating over $K_{ij}$, we have assumed that there are no restrictions on these variables. However, this is not generally true, since the matrix $K$ is positive-definite. Nevertheless, the conditions on the integration domain are lifted in the limit  $p\rightarrow\infty$.
    
    The action (\ref{FinalAction}) describes free particle motion on $\mathcal{F}_{1,2,3}$ with the metric
    \begin{gather}\label{F111metr}
        d s^2=
        \frac{\left|\Bar{u}_1\circ d u_2\right|^2}{\alpha}+
        \frac{\left|\Bar{u}_2\circ d u_3\right|^2}{\beta}+
        \frac{\left|\Bar{u}_1\circ d u_3\right|^2}{\gamma}.
    \end{gather}

    Thus, we have shown that the  spin chain~(\ref{spinChainHamiltonian}) leads to the free particle problem on the flag manifold as $p\rightarrow \infty$. Remarkably, our calculation provides an explicit relation between the metric coefficients in~(\ref{F111metr}) and the Hamiltonian coefficients in~(\ref{spinChainHamiltonian}). In other words, the following statement holds:
    \begin{prop}\label{spinFlagConnection}
        Let $\mathcal{H}$ be the Hamiltonian~(\ref{spinChainHamiltonian}) acting in the space~$V(p)$. Then
    \bea \label{partitionFuncF123}
        \underset{p\to\infty}{\mathrm{lim}}\,\mathrm{Tr}_{\mathrm{V}(p)}(e^{- \tau \mathcal{H}})=\mathrm{Tr}_{L^2(\mathcal{F}_{1,2,3})}(e^{- \tau\mathcal{H}_{\mathrm{particle}}})\,,
    \eea
    where $\tau>0$ is a parameter, $\mathcal{H}_{\mathrm{particle}}$ is the quantum mechanical Hamiltonian corresponding to classical action~(\ref{FinalAction}). In other words,  
    \bea
        \mathcal{H}_{\mathrm{particle}}=-\triangle\,,
    \eea
    where $\triangle$ is the Laplace-Beltrami operator for the metric~(\ref{F111metr}) on the flag manifold \cite[Section 3]{Wagner_2024}. 
    \end{prop} 

    In (\ref{partitionFuncF123}) we have chosen the additive constant in the Hamiltonian~$\mathcal{H}$ such that the ground state of the spin chain corresponds to the zero eigenvalue (see Section~\ref{flagMetr}). As a result, it follows from the equality of partition functions that the spectrum of $\mathcal{H}$ coincides with that of $-\triangle$ as $p\rightarrow \infty$ (see, for example, \cite[Chapter 3]{ZinnJustin:2004PathInt}).

\subsection{Case of $\gamma=\beta$.}\label{gammabetametrsec}

 In the following we restrict ourselves to the special case where $\gamma=\beta$. As the set $\{u_i\}_{i=1}^3$ forms an orthonormal basis in $\mathbb{C}^3$, it follows that there is a `partition of unity'
    \begin{gather}\label{partunit}
        u_1 \otimes \Bar{u}_1 + u_2 \otimes \Bar{u}_2 + u_3 \otimes \Bar{u}_3 = \Id\,.
    \end{gather}
    Thus, in this case the metric~(\ref{F111metr}) can be rewritten in the form
  \begin{gather}\label{alphabetametr}
        d s^2=
        \frac{1}{\alpha}\left|\Bar{u}_1\circ d u_2\right|^2+
        \frac{1}{\beta}\,
        du_3 \left(\Id-\Bar{u}_3\otimes u_3\right) d\Bar{u}_3\,.
    \end{gather}
    In particular, one can clearly see the structure of the forgetful map\footnote{See Section~\ref{forgetfulsec} below for more information on forgetful maps.} (bundle) 
\begin{gather}\label{F123forgetful}
\mathcal{F}_{1,2,3}\mapsto \CP^2\,,   
\end{gather}
defined by the map $(u_1, u_2, u_3)\mapsto u_3$. The terms in~(\ref{alphabetametr}) containing $u_3$ correspond to the Fubini-Study metric on $\CP^2$, whereas the first term in~(\ref{alphabetametr}) corresponds to the metric in the fiber $\CP^1$ parameterized by the vectors $(u_1, u_2)$ (orthogonal to $u_3$), see~(\ref{S2metr}). Therefore the size of the base and the size of the fiber are determined by $1\over \beta$ and $1\over \alpha$, respectively. Somewhat more formally, this means that the map~(\ref{F123forgetful}) is a Riemannian submersion (see, for example,~\cite{Cheeger}), when $\mathcal{F}_{1,2,3}$ is equipped with the metric~(\ref{alphabetametr})  (as opposed to the case of the general metric~(\ref{F111metr})).

Let us mention some special cases of the metrics~(\ref{alphabetametr}). First, there is the normal metric \cite{Wang1985}, which corresponds to the case $\alpha = \beta$. The case of $\beta = 2\alpha$ is another interesting possibility: this is the Kähler-Einstein metric\footnote{Apparently, Kähler-Einstein metric does not belong to the class of metrics described below with `nested' structure on $\mathcal{F}_{1,2,\ldots,N}, N\geq 4$ (see Section \ref{flagMetr}).} \cite{AlePer86, Achmed_Zade_2020}. Finally, in the limit $\frac{1}{\alpha} \to 0$ one arrives at the Fubini-Study metric on $\CP^2$ (which is, of course, K\"ahler and might be thought of as a degenerate  metric on $\mathcal{F}_{1,2,3}$).

 \subsection{Spectrum of the Hamiltonian for $\gamma = \beta$.}\label{eigenvaluesSpinChainSU(3)}
    Let us start by rewriting $\mathcal{H}$ in the more convenient form
    \begin{gather} \label{HamSU3}
        \mathcal{H}
        =
        \left(
            \alpha - \beta
        \right) S^a_1 S^a_2 + \beta 
        \left( 
            S^a_1 S^a_2 + S^a_2 S^a_3 + S^a_1 S^a_3 
        \right)+\mathrm{const}.  
    \end{gather}
    
    Here $S^a_i$ stand for the $\mathfrak{su}(3)$ generators, so that  $\left[S^a_i, S^b_i\right] = f^{abc}S^c_i$, where $f^{abc}$ are the structure constants; $a,b,c=1,\ldots,8$. One can easily show that
    \begin{gather}
        \left[
            S^a_1 S^a_2,
            S^b_1 S^b_2 + S^b_2 S^b_3 + S^b_1 S^b_3
        \right]
        =0.
    \end{gather}
    Note that 
    $
        S^a=
        S^a_1 + S^a_2 + S^a_3
    $
    are the $\mathfrak{su}(3)$ generators in the tensor product representation $\mathrm{V}(p)$ defined in~(\ref{Vpdef}), and 
    $
        S^a=
        S^a_1 + S^a_2
    $
    are the $\mathfrak{su}(3)$ generators in the representation $\mathrm{Sym}(p)^{\otimes 2}$. Thus, up to a factor of $\frac{1}{2}$ and terms proportional to the identity operator, $S^a_1 S^a_2 $ and $S^a_1 S^a_2 + S^a_2 S^a_3 + S^a_1 S^a_3$ are quadratic Casimir operators in $\mathrm{Sym}(p)^{\otimes 2}$ and in $\mathrm{V}(p)$, respectively.
    
    Let us describe the algorithm for finding the spectrum of $\mathcal{H}$. Initially, one can multiply the first two factors of  $\mathrm{Sym}(p)$ in $\mathrm{V}(p)$ and choose an irreducible representation, which we call $A$. This way we determine the eigenvalue of $S^a_1 S^a_2 $:  denote it by $\lambda_1$. Next, one can multiply $A$ by $\mathrm{Sym}(p)$ and choose an irreducible representation in the decomposition:  call it $B$. This determines the eigenvalue of $S^a_1S^a_2 + S^a_2 S^a_3 + S^a_1 S^a_3$: denote it by  $\lambda_2$. The corresponding eigenvalue of $\mathcal{H}$ is then $(\alpha - \beta) \lambda_1 + \beta \lambda_2$, its degeneracy being given by the dimension of $B$. To find the entire spectrum $\sigma(\mathcal{H},p)$, one should consider all possible pairs $(A, B)$.
    
    Let us illustrate this algorithm on the example of $p = 1$. In this case  
    \[
        \mathrm{V}(1) =
        \begin{ytableau}
               *(lime) ~ 
        \end{ytableau}
        ~\otimes~
        \begin{ytableau}
               *(lime) ~ 
        \end{ytableau}
        ~\otimes~
        \begin{ytableau}
               ~ 
        \end{ytableau}
        =
        \begin{ytableau}
              *(lime) ~ & *(lime) ~ & ~
        \end{ytableau}
        \oplus
        \begin{ytableau}
               *(lime)~ &*(lime) ~ \\
               ~ &\none 
        \end{ytableau}
        \oplus
        \begin{ytableau}
               *(lime)~ \\
              *(lime)~ \\
               ~
        \end{ytableau}
        \oplus
        \begin{ytableau}
               *(lime)~ & ~ \\
               *(lime)~ &\none 
        \end{ytableau}.
    \]
    It follows that every irreducible representation in $\mathrm{V}(1)$ corresponds to a certain irreducible representation in $\mathrm{Sym}(1)^{\otimes 2}$ (coloured green). The same holds for higher~$p$.
    
    Consider a Young diagram $Y(p_1, p_2, p_3)$ with row lengths $p_1\geq p_2\geq p_3\geq 0$. Introducing $s_i = p_i - \frac{1}{3} \sum_{j=1}^3 p_j$, one can calculate the value of the quadratic Casimir operator for the corresponding irreducible representation using the formula~\cite{PerelomovPopov} 
    \begin{gather}
        C_2(p_1,p_2,p_3)=\sum_{j=1}^3 s_j(s_j - 2j).
    \end{gather}

    We are thus in a position to write out an explicit formula for the eigenvalues of $\mathcal{H}$: 
    
    \begin{prop}
        The eigenvalues of $\mathcal{H}$, for a fixed value of $p$, take the form
        \begin{gather}\label{F123spectrum}
            \Lambda_k = \frac{\beta}{2}C_2(p^B_1,p^B_2,p^B_3)+
            \frac{\alpha - \beta}{2}
            \left[
                C_2\left(p_1^A, p^A_2, 0\right)-C_2(p,p,0)
            \right], 
        \end{gather}
        where $p_1^A+p_2^A=2p,\, p_2^A\leq p\leq p_1^A$ and $p_1^B+p_2^B+p_3^B=3p,\, p_2^B\leq p_1^A\leq p_1^B,\, p_3^B\leq p_2^A \leq p_2^B$.
    \end{prop}

    \begin{proof}
    According to the algorithm described above, we need to identify certain irreducible representations $A$ in $\mathrm{Sym}(p)^{\otimes 2}$ and $B$ in $V(p)$ such that $B $ is contained in the decomposition of $ A \otimes \mathrm{Sym}(p)$ into irreducibles. These representations correspond to certain Young diagrams $Y_A(p_1^A,p_2^A,p_3^A) \in \mathrm{Irr}\big(\mathrm{Sym}(p)^{\otimes 2}\big)$ and $Y_B(p_1^B,p_2^B,p_3^B) \in \mathrm{Irr}\big(V(p)\big)$, in terms of which the spectrum can be easily described. 

    The rules for multiplying by $\mathrm{Sym}(p)$ are quite simple. For instance,
    \begin{gather}\label{SymPsquared}
        \mathrm{Sym}(p)\otimes \mathrm{Sym}(p)=\underset{\mathrm{Irr}\big(\mathrm{Sym}(p)^{\otimes 2}\big)}{\bigoplus}\quad \begin{ytableau}
           ~ &  ~& ~& ~ & ~ \\
           ~ & ~ & ~
        \end{ytableau}\,,
    \end{gather}
    where the sum is taken over diagrams with row lengths $(p_1^A, p_2^A)$ satisfying the conditions
    \begin{gather}\label{Adiagcond}
        p_1^A+p_2^A=2p\,, \quad p_2^A\leq p\leq p_1^A.
    \end{gather}
    Take any diagram $Y_{A}$ on the r.h.s. of the equation~(\ref{SymPsquared}). Multiplying it by $\mathrm{Sym}(p)$, one gets diagrams $Y_{B}$ with three rows of lengths $(p_{1},p_{2},p_{3})$ subject to the conditions
    \begin{gather}\label{Bdiagcond}
    p_1^B+p_2^B+p_3^B=3p\,,\quad  p_2^B\leq p_1^A\leq p_1^B \,,\quad p_3^B\leq p_2^A \leq p_2^B.
    \end{gather}
    For each pair of diagrams $Y_A$ and $Y_B$, satisfying the conditions (\ref{Adiagcond})-(\ref{Bdiagcond}), the eigenvalue of the Hamiltonian can be calculated using the formula
    \begin{gather}
        \Lambda_k = \frac{\beta}{2}C_2(p^B_1,p^B_2,p^B_3)+
        \frac{\alpha - \beta}{2}
        \left[
            C_2\left(p_1^A, p^A_2, 0\right)-C_2(p,p,0)
        \right].
    \end{gather}
    As mentioned above, the degeneracy of $\Lambda_k$ is equal to the dimension of $B$. To find the entire spectrum, one needs to search for all possible values of the row lengths of $Y_A$ and~$Y_B$.
    \end{proof}

    Note that the additive constant in $\mathcal{H}$ has been chosen as $-\frac{\alpha-\beta}{2}C_2(p, p, 0)$. This is necessary to ensure that the eigenvalues have a limit as $p \rightarrow \infty$. Indeed, upon changing $p\mapsto p+1$ the Young diagrams $Y_A(p_1^A,p_2^A,0)\in \mathrm{Irr}\big(\mathrm{Sym}(p)^{\otimes 2}\big)$ and $Y_B(p_1^B,p_2^B, p_3^B)\in \mathrm{Irr}\big(V(p)\big)$ are mapped to the new diagrams $Y'_A(p_1^A+1,p_2^A+1, 0)\in \mathrm{Irr}\big(\mathrm{Sym}(p+1)^{\otimes 2}\big)$ and $Y'_B(p_1^B+1,p_2^B+1,p_3^B+1)\in \mathrm{Irr}\big(V(p+1)\big)$, respectively (see the end of Section~\ref{definitionOfTheSU(3)SpinChain}). It turns out that the values $\Lambda _k$ for $(Y_A, Y_B)$ are the same as those for $(Y'_A, Y'_B)$ (see Section~\ref{stabspectrum}), so that  $\sigma(\mathcal{H},p)\subset \sigma(\mathcal{H},p+1)$. As a result, in the limit $p \rightarrow \infty $ one finds the entire spectrum of the corresponding Laplace-Beltrami operator. Let us also mention that, with this particular choice of the additive constant, the energy of the ground state of the Hamiltonian is zero. For more on this see Section~\ref{SpinChain}.

\subsubsection{Reduction to $\CP^2$.}\label{reductionToCP2} 
The limit $\alpha \to \infty$ is also of interest. In this case, as discussed in Section~\ref{gammabetametrsec}, the metric on $\mathcal{F}_{1, 2, 3}$ degenerates. As the fiber size tends to zero, the metric converges to a metric on the base, that is, on $\mathbb{CP}^2$. It follows from the formula~(\ref{F123spectrum}) for the spectrum and from Appendix~\ref{groundState} that the eigenvalues, which are bounded as $\alpha\to \infty$, correspond to the states for which $C_2\left(p_1^A, p^A_2, 0\right)=C_2(p,p,0)$ holds. Thus, out of the set of Young diagrams $A$ only rectangular diagrams of the following form should remain in the limit:
\begin{gather}\label{2pdiag}
        \underbrace{
            \begin{ytableau}
           ~ & ~ & \dots & ~ & ~ \\
           ~ & ~ & \dots & ~ & ~
        \end{ytableau}
    }_{p}.
\end{gather}
The corresponding representations are symmetric powers of the representation conjugate  to the fundamental, i.e. $\widebar{\mathrm{Sym}(p)}$. In the limit $p\to \infty$ we obtain the correct identification of Hilbert spaces:
    \begin{gather}\label{CP2Hilb}
        L^2\left(\CP^2\right)=\,\underset{p\to\infty}{\mathrm{lim}}\, \mathrm{Sym}(p)\otimes \widebar{\mathrm{Sym}(p)}\,.
    \end{gather}
    In this case the analogue of Proposition~\ref{lagrtheorem} is that there  exists a Lagrangian embedding $\CP^2 \to \CP^2\times\CP^2$ given by the map $z \to (z,\bar{z})$. Complex conjugation in the second factor intuitively corresponds to taking the  conjugate representation in~(\ref{CP2Hilb}).

    In fact, the system on $\mathcal{F}_{1,2,3}$ becomes equivalent to a system on the partial flag manifold $\mathcal{F}_{2,3} \cong \mathbb{CP}^2$ as $\alpha \rightarrow \infty$. Indeed, the rectangular Young diagram~(\ref{2pdiag}) leads to coherent states parameterized by points on the Grassmannian  $\mathcal{F}_{2,3}$. This means that, in this limit, the pair of vectors $(u_1, u_2)$ describing $\mathcal{F}_{1,2,3}$ is only defined up to a  $\mathrm{U}(2)$ gauge transformation. As a result, the gauge group is enlarged from $\mathrm{U}(1)\times \mathrm{U}(1)$ to~$\mathrm{U}(2)$.
    
    \subsection{Geodesics in the case of $\gamma=\beta$.}
    Previously, we derived the spectrum of the quantum Hamiltonian~(\ref{spinChainHamiltonian}) for the case $\beta=\gamma$, thereby solving the quantum mechanical problem on $\mathcal{F}_{1,2,3}$. The classical problem corresponding to the action~(\ref{FinalAction}) is equivalent to the geodesic problem on this manifold. 
    Let us state the main proposition:
\begin{prop}
       The general solution to the geodesic equations on $\mathcal{F}_{1,2,3}$ with the metric~(\ref{alphabetametr}) has the form
        \begin{gather}\label{flag111sol}
        \begin{pmatrix}
            u_1(t) \\
            u_2(t) \\
            u_3(t)
        \end{pmatrix}
        =
        G(t)
        \exp
        \left[
        \begin{pmatrix}
            0 & \frac{1}{\alpha}A^0_{21} & \frac{1}{\beta}A^0_{31}\\
            -\frac{1}{\alpha}\Bar{A}^0_{21} & 0 & \frac{1}{\beta}A^0_{32}\\
            -\frac{1}{\beta}\Bar{A}^0_{31} & -\frac{1}{\beta}\Bar{A}^0_{32} & 0
        \end{pmatrix}\beta t    
        \right]
        \begin{pmatrix}
            u_1(0) \\
            u_2(0) \\
            u_3(0)
        \end{pmatrix}\,,
    \end{gather}
    where
    \begin{gather}\label{Gmatr}
        G(t):=
        \begin{pmatrix}
            \exp \left[
            \begin{pmatrix}
                0 & \frac{1}{\alpha}A^0_{21}\\
                -\frac{1}{\alpha}\Bar{A}^0_{21} & 0
            \end{pmatrix}(\alpha - \beta)t
            \right] & \\
            & 1
        \end{pmatrix}
    \end{gather}
    and $A_{21}^0, A_{31}^0, A_{32}^0$ are the initial data.
    
    \end{prop}

    \begin{proof}
        
    Our model's action takes the form~(\ref{FinalAction}) with $\beta = \gamma$. 
    The Lagrangian can be rewritten by analogy with the metric~(\ref{alphabetametr}) as
    \begin{gather}\label{alphabetaLagr}
        \mathcal{L}=
        \frac{1}{\alpha} \dot{\Bar{u}}_2\circ (u_1\otimes \Bar{u}_1)\circ \dot{u}_2 
        + 
        \frac{1}{\beta} \dot{\Bar{u}}_3\circ (\Id - u_3\otimes \Bar{u}_3)\circ \dot{u}_3
        +
        \lambda^{ij}(\Bar{u}_i\circ u_j - \delta_{ij})\,,
    \end{gather}
    where $\lambda^{ij}$ are Lagrange multipliers imposing orthonormality of the vectors~$u_k$.

The flag manifold $\mathcal{F}_{1,2,3} $ is the quotient space $ \frac{\mathrm{U}(3)}{\mathrm{U}(1) \times \mathrm{U}(1) \times \mathrm{U}(1)}$, and, as a result, the model~(\ref{alphabetaLagr}) is invariant under the action of the gauge group $\mathrm{U}(1)^3$. Gauge transformations take the form $u_i\rightarrow e^{i\varphi_i(t)}u_i$. The natural gauge conditions that we use are $\Bar{u}_i \circ \dot{u}_i = 0\,, i=1, 2, 3$.

Let us introduce the following notation: $A_{ij}(t):= \Bar{u}_i \circ \dot{u}_j(t)$, $A^0_{ij}:= \Bar{u}_i \circ \dot{u}_j(0)$. Due to the orthonormality of the vectors $u_i$ we have $A_{ij}=-\Bar{A}_{ji}$. As in the $\mathrm{SU}(2)$ case (see~(\ref{S2EoM})), the equations of motion lead to equations on $A_{ij}(t)$ (for more details see Appendix~\ref{su3flagEOM}):
\begin{gather}
    \frac{d}{dt} A_{21}(t)=0\,,\quad 
      \frac{d}{dt}
        \begin{pmatrix}
            A_{31} \\
            A_{32}
        \end{pmatrix}
        =
        \begin{pmatrix}
            0 & -\beta \left(
            \frac{1}{\alpha} - \frac{1}{\beta}
        \right) A_{21} \\
            -\beta \left(
            \frac{1}{\alpha} - \frac{1}{\beta}
        \right) A_{12}  & 0
        \end{pmatrix}
        \begin{pmatrix}
            A_{31} \\
            A_{32}
        \end{pmatrix}
        \,, \label{EvolutionEquation}
\end{gather}
  with the solution
    \begin{gather}
      A_{21} (t) = \mathrm{const} = A^0_{21}\,,\quad\quad 
        \begin{pmatrix}
            A_{31}\\
            A_{32}
        \end{pmatrix} = 
        \exp \left[
        \begin{pmatrix}
            0 & \frac{1}{\alpha}A^0_{21}\\
            -\frac{1}{\alpha}\Bar{A}^0_{21} & 0
        \end{pmatrix}(\alpha - \beta)t
        \right]
        \begin{pmatrix}
            A^0_{31}\\
            A^0_{32}
        \end{pmatrix}.
    \end{gather}
        
    Knowing $A_{ij}$, one can find the evolution equations of $u_i$, as the vectors $\{u_i\}_{i=1}^{3}$ form an orthonormal basis in $\mathbb{C}^{3}$ at every instant of time. From the `partition of unity'~(\ref{partunit}) it follows that
    \begin{gather}\label{ueqs}
        \frac{d}{dt}
        \begin{pmatrix}
            u_1 \\
            u_2 \\
            u_3
        \end{pmatrix}
        =
        \begin{pmatrix}
            0 & A_{21} & A_{31} \\
            -\Bar{A}_{21} & 0 & A_{32}\\
            -\Bar{A}_{31} & -\Bar{A}_{32} & 0 
        \end{pmatrix}
        \begin{pmatrix}
            u_1 \\
            u_2 \\
            u_3
        \end{pmatrix}
        :=
        \mathcal{A}(t)
        \begin{pmatrix}
            u_1 \\
            u_2 \\
            u_3
        \end{pmatrix}.
    \end{gather}
   
    %

    The formal solution can be written in the form of an ordered exponential, but it is also possible to obtain an explicit expression. To this end we need to recall the gauge invariance of the system of equations~(\ref{ueqs}). Consider the equation $\frac{d}{dt}x=Mx$ and make the substitution $x=G(t)y$, where $G(t)$ is a unitary transformation. This gives $\frac{d}{dt}y=M'y$, where $M'=G^{\dagger}MG-G^{\dagger}\frac{d}{dt}G$.

    Let us take the matrix~(\ref{Gmatr}) for $G$. In this case the equation becomes
    \begin{gather}
        \frac{d}{dt}
        \left(
            G^{\dagger}
            \begin{pmatrix}
                u_1 \\
                u_2 \\
                u_3
            \end{pmatrix}
        \right)
        =
        \beta
        \begin{pmatrix}
            0 & \frac{1}{\alpha}A^0_{21} & \frac{1}{\beta}A^0_{31}\\
            -\frac{1}{\alpha}\Bar{A}^0_{21} & 0 & \frac{1}{\beta}A^0_{32}\\
            -\frac{1}{\beta}\Bar{A}^0_{31} & -\frac{1}{\beta}\Bar{A}^0_{32} & 0
        \end{pmatrix}
        G^{\dagger}
        \begin{pmatrix}
            u_1 \\
            u_2 \\
            u_3
        \end{pmatrix}.
    \end{gather}
    Therefore the solution has the form~(\ref{flag111sol}).

    \end{proof}


\subsubsection{Special cases.}
    Having the explicit form of the solution at hand, one can study several interesting special cases.

    \begin{itemize}
        \item $A_{31}^0=A_{32}^0=0$. In this case 
         \begin{gather}
        \begin{pmatrix}
            u_1(t) \\
            u_2(t) \\
            u_3(t)
        \end{pmatrix}
        =
        \exp
        \left[
        \begin{pmatrix}
            0 & A^0_{21} & 0\\
            -\Bar{A}^0_{21} & 0 & 0\\
            0 & 0 & 0
        \end{pmatrix} t    
        \right]
        \begin{pmatrix}
            u_1(0) \\
            u_2(0) \\
            u_3(0)
        \end{pmatrix}\,.
    \end{gather}
    It can be seen that the solution is a rotation in the $(u_1, u_2)$ plane. In other words, it is a great circle on the sphere $\CP^1$, i.e. in the fiber of the forgetful bundle~(\ref{F123forgetful}). It follows that the fibers  are totally geodesic submanifolds.

    \item $A_{21}^0=0$. Then
    \begin{gather}
        \begin{pmatrix}
            u_1(t) \\
            u_2(t) \\
            u_3(t)
        \end{pmatrix}
        =
        \exp
        \left[
        \begin{pmatrix}
            0 & 0 & A^0_{31}\\
            0 & 0 & A^0_{32}\\
            -\Bar{A}^0_{31} & -\Bar{A}^0_{32} & 0
        \end{pmatrix} t    
        \right]
        \begin{pmatrix}
            u_1(0) \\
            u_2(0) \\
            u_3(0)
        \end{pmatrix}\,.
    \end{gather}
    This is the horizontal lift to $\mathcal{F}_{1, 2, 3}$ (satisfying the condition $\Bar{u}_2 \circ \dot{u}_1 = 0$) of a geodesic from the base $\CP^2$. It is a well-known fact that, in the case of a Riemannian submersion, the horizontal lift of a geodesic from the base space is a geodesic in the total space (see~\cite[Theorem 3.31]{Cheeger})
\end{itemize}
 The explicit solution~(\ref{flag111sol})-(\ref{Gmatr}) provides the description of arbitrary geodesics, not necessarily `vertical' or `horizontal' ones. Let us also discuss special cases of metrics for which the solution can be simplified:
 
\begin{itemize}

    \item $\alpha=\beta$ (the normal metric). As can be seen from~(\ref{flag111sol}), in this case all  geodesics are orbits of one-parameter isometry groups. This is a general property of normal metrics~\cite{Cheeger}.

    \item $\alpha\rightarrow\infty$ (the Fubini-Study metric on $\CP^2$). This limit has already been studied in the context of the spin chain at the end of Section~\ref{eigenvaluesSpinChainSU(3)}. In this case we obtain
    \begin{gather}
        \begin{pmatrix}
            u_1 \\
            u_2 \\
            u_3
        \end{pmatrix}
        =
        \exp
        \left[
        \begin{pmatrix}
            0 & A^0_{21} & 0\\
            -\Bar{A}^0_{21} & 0 & 0\\
            0 & 0 & 0
        \end{pmatrix} t    
        \right]
        \exp
        \left[
        \begin{pmatrix}
            0 & 0 & A^0_{31}\\
            0 & 0 & A^0_{32}\\
            -\Bar{A}^0_{31} & -\Bar{A}^0_{32} & 0
        \end{pmatrix} t    
        \right]
        \begin{pmatrix}
            u_1 \\
            u_2 \\
            u_3
        \end{pmatrix}(0)\,.\label{alphaLimitGeod}
    \end{gather}
    As noted above, the vectors $(u_1,u_2)$ are defined up to an action of $\mathrm{U}(2)$. Accordingly, up to an additional gauge transformation, (\ref{alphaLimitGeod}) is a geodesic in $\mathbb{CP}^2$  (again, this is an orbit of a one-parameter subgroup of $\mathrm{SU}(3)$).
    
    \end{itemize}

    
    %
    %
    %
\section{Generalization to higher $N$}\label{Nsec}
    This section provides a natural generalization of the constructions described in Section \ref{SU3Case} to the $\SU(N)$ case, where $N \geq 3$.

    \subsection{Spin chain Hamiltonian.} \label{SpinChain}
    First of all, one should note a basic property of the Hamiltonian (\ref{HamSU3}), namely that it can be written as a sum of two commuting Casimir operators. One of these operators acts in $\mathrm{Sym}(p)^{\otimes 2}$, while the other  acts in $\mathrm{Sym}(p)^{\otimes 3}$. In the following we use this observation to construct the generalization.
    
    Consider a system of $N$ $\SU(N)$-spins. In order to describe the Hamiltonians (as well as the metrics on flag manifolds later on) it is convenient to use the graphical notation

    \begin{equation}
        \begin{overpic}[scale=1.2,unit=0.5mm]{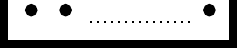}
            \put(10,22){$I$}
            \put(20,22){$I+1$}
            \put(75,22){$I+J$}
        \end{overpic}
    \end{equation}
    \noindent where each node corresponds to a `copy' of $\mathrm{Sym}(p,N)$, the $p$-th symmetric power of the  fundamental representation of $\SU(N)$, representing a spin at that site. The \textit{bracket} denotes the addition to the Hamiltonian of the quadratic Casimir operator constructed on these spins, taken with some multiplier. In turn, we refer to this multiplier as the \textit{coefficient of the bracket}. For example, the above picture is represented by the operator
    \begin{gather}\label{HamIJdef}
        \alpha_{I,J} H_{I,J}:=\alpha_{I,J}\left[\frac{1}{2}\left(\sum_{i=I}^{I+J} S^a_i\right)\left(\sum_{j=I}^{I+J} S^a_j\right)+\textrm{const}\right]\,,
    \end{gather}
    where $S^a=\sum_{i=I}^{I+J} S^a_i$ is a generator of $\mathfrak{su}(N)$ taken in the representation $\mathrm{Sym}(p,N)^{\otimes (J+1)}$. For convenience we have also included an additive constant in the definition of the Hamiltonian, whose significance will be discussed later on.

      We will be considering Hamiltonians whose diagrams do not contain partially overlapping brackets. For example, the following configuration is \textit{forbidden}:

    \begin{equation}
        \begin{overpic}[scale=1.2,unit=0.5mm]{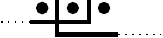}
        \end{overpic}
    \end{equation}

    In other words, the brackets, from which the Hamiltonian $\mathcal{H}$ is built, must be `embedded' into each other. In this case the corresponding quadratic Casimir operators $H_{I,J}$ commute with each other. We also note that $\mathcal{H}$ acts in the space $\mathrm{Sym}(p,N)^{\otimes N}$ by construction.

    We also impose the following condition on the coefficients $\alpha_{I,J}$ of the brackets: the coefficient in front of each $S^a_i S^a_j$ term\footnote{Which is the sum of coefficients of all brackets that contain $S^a_i S^a_j$.} in $\mathcal{H}$ must be positive. This convention is explained in Section \ref{flagMetr}.
    
 \subsubsection{Examples.}    
    Using this notation, we may represent the Hamiltonian (\ref{HamSU3}) of the $\SU(3)$ spin chain  by the diagram
    \begin{equation}
        \begin{overpic}[scale=1.2,unit=0.5mm]{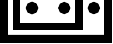}
            \put(11,19){$1$}
            \put(24,19){$2$}
            \put(37,19){$3$}
        \end{overpic}
    \end{equation}
    
    \noindent The first bracket has the coefficient $\alpha-\beta$, whereas the second bracket has the coefficient~$\beta$. 

    As another illustration, consider an example of six spins with the diagram
    \begin{equation}
        \begin{overpic}[scale=1.2,unit=0.5mm]{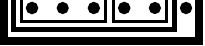}
            \put(11,20){$1$}
            \put(23,20){$2$}
            \put(36,20){$3$}
            \put(49,20){$4$}
            \put(61,20){$5$}
            \put(74,20){$6$}
        \end{overpic}
        \label{diag16}
    \end{equation}
    
    \noindent In this case the Hamiltonian takes the form
    \begin{gather}
        \mathcal{H}=\alpha_{1,3} H_{1,3} + \alpha_{4,5} H_{4,5}+\alpha_{1,5} H_{1,5} + 
        \alpha_{1,6} H_{1,6}.
    \end{gather}

    \subsubsection{Spectrum of the Hamiltonian.} 
    The algorithm for calculating the eigenvalues of $\mathcal{H}$ for a particular $p$ is as follows:
    \begin{enumerate}
        \item 
        Find the brackets that contain no smaller ones. Choose one of them and consider the corresponding term $H_{I,J}$.
         Multiply the representations $\mathrm{Sym}(p,N)$ corresponding to the given bracket and choose an irreducible representation in the decomposition\footnote{It is convenient to use Young diagrams for this purpose. Recall that $\mathrm{Sym}(p, N)$ is matched by a diagram with $p$ boxes arranged in a single row.
            }. The eigenvalue of $\alpha_{I,J}H_{I,J}$ is given by 
            \begin{gather}\label{eigenvalueSU(N)}
                \frac{\alpha_{I,J}}{2}\left(C_2(p_1,\ldots,p_N)-C_2(p,\ldots,p_{J+1}=p,0,\ldots,0)\right)\,,
            \end{gather}
            where $p_j$ are the lengths of rows of the diagram (it is clear that $p_j = 0$ for $j > (J+1)$) and 
\begin{gather}\label{quadraticCasimirSU(N)}
        C_2(p_1,\ldots,p_N)=\sum_{j=1}^N s_j(s_j - 2j),\quad\quad\textrm{with}\quad\quad s_i=p_i-\frac{1}{N}\sum_{j=1}^N p_j \,,
            \end{gather}
            is the eigenvalue of the quadratic Casimir operator of $\SU(N)$  \cite{PerelomovPopov}. The subtraction in~(\ref{eigenvalueSU(N)}) makes use of the possibility of adding a constant to the Hamiltonian~(\ref{HamIJdef}). It turns out that, with this choice, the ground state has eigenvalue zero (see  Section~\ref{stabspectrum} below and Appendix \ref{groundState}). For each bracket we use this method to find the corresponding eigenvalues.
        \item
        Find the brackets containing no smaller ones, this time ignoring the ones from  step 1. Let us take a look at one of these. Some of the brackets from step 1 fall inside this bracket, and at step 1 we have picked from each such smaller bracket a Young diagram corresponding to one of the irreducible components. Tensor-multiply all of these irreducible representations, as well as  with the $\mathrm{Sym}(p,N)$ representations  at the nodes that did not fall into any of the brackets from step~1 but  fall into the new bracket. Again, choose a  Young diagram corresponding to an irreducible component and find its eigenvalue using~(\ref{eigenvalueSU(N)}). For each bracket found at this step, add the obtained values to the result of step 1.

        \item Continue until all the brackets are exhausted.
    \end{enumerate}
    
    At each step, a particular irreducible component is picked from the tensor product of representations determined at the previous step (and, possibly, $\mathrm{Sym}(p,N)$). In order to find the whole spectrum of $\mathcal{H}$, one should consider all possible combinations.

    \subsubsection{Spectrum stabilization.} \label{stabspectrum}
    In the following, we define the Hamiltonian $\mathcal{H}$ by taking into account the subtraction introduced in  equation (\ref{eigenvalueSU(N)}). To explain this, let us analyze the behavior of the spectrum $\sigma(\mathcal{H},p)$ of the Hamiltonian as $p$ increases. Similarly to the case of $\SU(3)$, after the transition $p\mapsto p+1$ a Young diagram of an irreducible representation in $\mathrm{Sym}(p,N)^{J+1}$ maps to a Young diagram with an additional column of $J+1$ boxes, i.e. to a certain irreducible component in $\mathrm{Sym}(p+1,N)^{J+1}$. In this case the eigenvalues of $H_{I,J}$ in $\mathcal{H}$ increase, but the differences between them remain unaltered. Indeed, let us show that 
     \begin{gather}
       C_2(p_1,\ldots,p_{J+1},0\ldots,0)-C_2(p, \ldots, p,0,\ldots,0)
       =\nonumber\\
       =C_2(p_1+1,\ldots,p_{J+1}+1,0\ldots,0)-C_2(p+1, \ldots, p+1,0,\ldots,0). \label{limitCheck}
    \end{gather}
   Subtracting the r.h.s. of equation (\ref{limitCheck}) from the l.h.s., upon some simplification we get
    \begin{gather}
        2\left(1-\frac{J+1}{N}\right)\sum^{J+1}_{i=1}(p_i - p)=0.
    \end{gather}
    This holds true, since, according to the rules for multiplying Young diagrams,  ${\sum\limits_{i=1}^{J+1} p_i = (J+1)p}$.

    However, the ground state energy is always zero due to the subtraction in~(\ref{eigenvalueSU(N)}) (see Appendix~\ref{groundState}). Thus, the spectrum stabilises, i.e.
    \begin{gather}
        \sigma(\mathcal{H},p) \subset \sigma(\mathcal{H},p+1).
    \end{gather}
    Specifically, the eigenvalues are essentially independent of $p$ (for sufficiently large $p$).   
    
    \subsection{Metrics on $\mathcal{F}_{1,2,\ldots,N}$.}\label{flagMetr}
    As with $\mathcal{F}_{1,2,3}$, there is a connection between an $\SU(N)$ spin chain in the limit $p\rightarrow \infty$ and a free-particle problem on $\mathcal{F}_{1,2,\ldots,N}$.
    
    First, let us recall the most general metric on a complete flag manifold (see also~\cite{Arvanitoyeorgos_1993}). This is a natural generalization of~(\ref{F111metr}):
    \begin{gather}\label{genmetr}
        d s^2=\sum_{i\neq j}^N {1\over \alpha_{ij}} |\Bar{u}_i \circ d u_j|^2.
    \end{gather}

    This metric is non-degenerate under the condition ${1\over \alpha_{ij}}>0$ for all $i, j$. Indeed, since we are dealing with a homogeneous space, it is sufficient to study the corresponding quadratic form near any point, e.g. $U:=(u_1, \ldots, u_N)=\mathrm{Id}\in \mathcal{F}_{1,2,\ldots,N}$. Using the decomposition $U=\mathrm{Id}+i\,H+\ldots$, with $H$ Hermitian, we have $\Bar{u}_j \circ d u_k=i\,dH_{jk}+\ldots$ to the leading order. So $ds^2\big|_{U=\mathds{1}_N}=\sum_{i\neq j}^N {1\over \alpha_{ij}}\,|dH_{ij}|^2$, so that the requirement ${1\over \alpha_{ij}}>0$ is obvious.  
    
    The recipe for deriving a metric from a spin chain Hamiltonian is the same as that described in Section~\ref{MetricsFlag}. Let us formulate a generalization of Proposition~\ref{spinFlagConnection}:

    \begin{prop}
        Let $\mathcal H$ be the Hamiltonian of an $\SU(N)$ spin chain with $N$ spins, acting in the space~$\mathrm{V}(p,N):=\mathrm{Sym}(p,N)^{\otimes N}$. Then
    \bea \label{HamLaplaceLim}
        \underset{p\to\infty}{\mathrm{lim}}\,\mathrm{Tr}_{\mathrm{V}(p,N)}(e^{- \tau \mathcal{H}})=\mathrm{Tr}_{L^2(\mathcal{F}_{1,2,\ldots,N})}(e^{- \tau\mathcal{H}_{\mathrm{particle}}})\,,
    \eea
    where $\tau>0$ is a parameter, $\mathcal{H} = -\triangle$ is the quantum mechanical Hamiltonian for a free particle on $\mathcal{F}_{1,2,\ldots,N}$, and $\triangle$ the Laplace-Beltrami operator corresponding to the metric constructed from $\mathcal{H}$.
    \end{prop} 
     
    It follows from the equality of partition functions that the spectrum of $\mathcal{H}$ coincides with that of $-\triangle$ as $p\rightarrow \infty$. 
    
   Let us explain why the subtraction introduced earlier in~(\ref{eigenvalueSU(N)}) is also necessary for the validity of~(\ref{HamLaplaceLim}). Obviously, $-\triangle$ is a non-negative operator with zero eigenvalue corresponding to constant functions. On the other hand, as discussed above, the Hamiltonians $H_{I, J}$ with the chosen subtraction are also non-negatively defined, and, as follows from~(\ref{eigenvalueSU(N)}), their zero eigenvalues correspond to rectangular Young diagrams of size $p\times(J-I+1)$.

   By successively considering all brackets at each step, we see that the zero eigenvalue corresponds to a rectangular diagram of width $p$ and increasing height. At the last step we arrive at the diagram of size $p \times N$, which is the singlet (the trivial representation). It is clear that the singlet corresponds to constant functions as $p \to \infty$.
    
    One can formulate the following statement based on the results of Section~\ref{SpinChain}:
    
    \begin{prop} \label{LaplaceBeltramiSpect}
        The spectrum of the Laplace-Beltrami operator $\triangle$, corresponding to the Hamiltonian $\mathcal{H}$ under consideration, can be described explicitly.
    \end{prop}

    In order to obtain the metric from the Hamiltonian, it suffices to replace the term $S^a_iS^a_j$ in the Hamiltonian~$\mathcal{H}$ with $|\Bar{u}_i \circ d u_j|^2$ and  all of the coefficients in~$\mathcal{H}$ with their inverse values. This is the same process that was used for $\mathcal{F}_{1,2,3}$ (see the formulas (\ref{spinChainHamiltonian}) and (\ref{F111metr})). It is clear that the metric obtained through this process has a `nested' structure inherited from the original Hamiltonian. Therefore one can use the graphical representation described in Section \ref{SpinChain}. The corresponding diagram

    \vspace{0.3cm}
    \begin{equation}
        \begin{overpic}[scale=1.2,unit=0.5mm]{Element.pdf}
            \put(10,22){$I$}
            \put(20,22){$I+1$}
            \put(75,22){$I+J$}
        \end{overpic}
    \label{diagIJ}
    \end{equation}
    \noindent 
    will again be called a \textit{bracket}, as in Section~\ref{SpinChain}.  
    
    The bracket above corresponds to the term $\frac{1}{\Tilde{\alpha}_{I,J}}\sum_{i=I+1}^{I+J}\sum_{j=I}^{i-1}|\Bar{u}_j \circ d u_i|^2$ in the metric, where the orthonormal set of vectors $\{u_k\}_{k=1}^N $ parameterizes the flag under consideration, as in the case of $\mathcal{F}_{1,2,3}$. Each node, which previously had its own spin, now has its own vector $u_i$. We refer to $\frac{1}{\Tilde{\alpha}_{I,J}}$ as the \textit{coefficient of the bracket}. One can relate these coefficients to the $\alpha_{I,J}$ coefficients in $\mathcal{H}$. As noted above, the coefficients in front of $S^a_iS^a_j$ and $|\Bar{u}_i \circ d u_j|^2$ are inverses of each other, but one should bear in mind that the interaction of the form $S^a_i S^a_j$ may fall into more than one bracket. In this case the final coefficient for this interaction is the sum of coefficients of all such brackets.

    From the non-degeneracy of the metric it follows that the coefficient in front of each term $|\Bar{u}_i \circ d u_j|^2$ must be positive. Since it is the inverse of the coefficient in front of $S^a_i S^a_j$ in $\mathcal{H}$, this 
    explains the positivity conditions introduced in Section~\ref{SpinChain}.

    As in the case of Hamiltonians, configurations with partially overlapping brackets are forbidden (this assumption is made throughout the rest of this paper). We refer to such metrics as having `nested' structure. In the following section we will establish their relation to forgetful maps of flag manifolds.

    \subsection{Forgetful bundle.}
\label{forgetfulsec}

    Let us briefly recall the construction of a \textit{forgetful} bundle on $\mathcal{F}_{1,2,\ldots,N}$ \cite{Donagi_2008, Manin1988GaugeFT}. It is given by the projection 
    \begin{gather}\label{forgetfulBundle}
        \mathcal{F}_{1,2,\ldots,N} \xrightarrow{\mathcal{F}_{1,2,\ldots,K}} \mathcal{F}_{1,2,\ldots,M,M+K,\ldots,N}.
    \end{gather}
    Speaking informally, we `forget' part of the fine structure of the flag, i.e. subspaces of certain dimensions. In~(\ref{forgetfulBundle}), these are subspaces of dimensions $M+1,M+2,\ldots,M+K-1$. The fiber of the bundle is $\mathcal{F}_{1,2,\ldots,K}$. In terms of the set of vectors $\{u\}_{k=1}^N$ parametrizing a point in $\mathcal{F}_{1,2,\ldots,N}$, we simply choose a subset of $K$ vectors and consider these to be defined up to an $\mathrm{U}(K)$ transformation (since $\{u_{k}\}_{k= 1}^{N}$ are orthonormalized). In other words, we `remember' only the plane spanned by these $K$ vectors.

    We can as well forget another set of dimensions: 
    \begin{gather} \label{towerForgetfulBundle}
        \mathcal{F}_{1,2,\ldots,N} \xrightarrow{\mathcal{F}_{1,2,\ldots,K}} \mathcal{F}_{1,2,\ldots,M,M+K,\ldots,N}\xrightarrow{\mathcal{F}_{1,2,\ldots,L}}\mathcal{F}_{1,2,\ldots,P,P+L,\ldots,M,M+K,\ldots,N}.
    \end{gather}
    \noindent Any partial flag manifold can be obtained using this method.




    Now we can define more precisely the notion of metric with `nested' structure. Each bracket of the form~(\ref{diagIJ}) in a metric diagram (regardless of its internal structure) corresponds to a forgetful map with fiber $\mathcal{F}_{1, 2, \ldots, J-I+1}$. If the bracket has internal structure, it means that the manifold $\mathcal{F}_{1, 2, \ldots, J-I+1}$ itself can be represented as an analogous bundle\footnote{For example, the diagram~(\ref{diag16}) means that $\mathcal{F}_{1, \ldots, 6}$ can be viewed as a bundle over $\CP^5\simeq \mathcal{F}_{5, 6}$ with fiber $\mathcal{F}_{1, \ldots, 5}$. Additionally, $\mathcal{F}_{1, \ldots, 5}$ may be viewed as a bundle over $\mathcal{F}_{2,5}\simeq\mathcal{F}_{3,5}$, with fiber $\mathcal{F}_{1,2,3}\times\mathcal{F}_{1,2}$.}. The `nested structure' of a metric is marked by the fact that each forgetful map is a \emph{Riemannian submersion}. As it turns out, in this case  fibers of the corresponding projections turn out to be totally geodesic submanifolds (see the~proof of Proposition~\ref{geodsolproof}). 
    
    In relation to forgetful maps, let us recall the following result~\cite{Bourguignon} (see also~\cite{Besson}):
    \begin{prop}[\cite{Bourguignon}]
        Let 
    \begin{gather}\label{fiberEB}
        (E_{N+M}, G) \xrightarrow{F_M}  (B_N, g)
    \end{gather} be a bundle over the base $B_N$, with fiber $F_M$ (the indices denote dimensions of these manifolds). The total space $E_{N+M}$ and the base $B_N$ are equipped with metrics $G$ and~$g$, respectively. Suppose additionally that the  map~(\ref{fiberEB}) is a Riemannian submersion with totally geodesic fibers. Then the Laplace-Beltrami operator on $E_{N+M}$ can be decomposed into a sum of `horizontal' and `vertical' Laplacians. These operators commute with each other.
    \end{prop}

\begin{proof}
    Let us pick coordinates $x^1,\ldots,x^N$ on $B_N$ and  $y^1,\dots,y^M$ on $F_M$. Furthermore, let $(ds^2)_B=g_{AB}(x)dx^Adx^B$ denote the metric on the base, and let the following be the metric on the total space:
    \begin{gather}
    (ds^2)_E=(ds^2)_B+h_{\alpha\beta}(y)\,\left(dy^\alpha-\mathcal{A}^{\alpha}_A(x) dx^A\right)\left(dy^\beta-\mathcal{A}^{\beta}_B(x) dx^B\right)\,.
    \end{gather}
    In this case the map~(\ref{fiberEB}) is a Riemannian submersion. One can verify by explicit calculation that the Laplacian constructed using this metric has the form (here ${g:=\mathrm{Det}(g_{AB})}$ and $h:=\mathrm{Det}(h_{\alpha\beta})$)  
    \begin{gather}\label{ELaplacian}
        \triangle_E={1\over g^{1/2}h^{1/2}}D_A\left(g^{1/2}h^{1/2} g^{AB} D_B\right)+{1\over h^{1/2}}\dd_\alpha\left(h^{1/2} h^{\alpha\beta} \dd_\beta\right)\,,\\
        \textrm{where}\quad\quad D_A:=\dd_A+\mathcal{A}_A^\alpha \dd_\alpha.
    \end{gather}

    The requirement that the fiber $F_M$ be totally geodesic  can be expressed as the condition\footnote{A well-known example of a bundle with totally geodesic fibers is the Hopf fibration $S^{2n+1}\to \CP^{n}$. In this case $M = 1$ and $h_{11} (y)=\mathrm{const}$, so that condition~(\ref{totgeod}) is satisfied.}
    \begin{gather}\label{totgeod}
        \Gamma^A_{\beta\gamma}=-{1\over 2} g^{AB}\mathcal{A}_B^\alpha \dd_\alpha h_{\beta\gamma}=0\,.
    \end{gather}
    In this case the two terms in~(\ref{ELaplacian}) (the so-called `horizontal' and `vertical' Laplace operators, respectively~\cite{Bourguignon}) evidently commute with each other. 
\end{proof}

In the case of forgetful bundles this proposition is a complete analogue of our statement about the commutativity of quantum Hamiltonians in the limit $p \to \infty$.

    \subsection{Geodesics.}\label{geodesicsec} We have described the relation between an $\SU(N)$ spin chain and a free-particle problem on the manifold of complete flags. Previously, we have developed an algorithm for calculating the eigenvalues of spin chain Hamiltonians and of the  corresponding Laplace–Beltrami operators for metrics with `nested' structure. This resolves the quantum version of the free-particle problem on $\mathcal{F}_{1,2,\ldots,N}$. What remains is to solve the classical problem, i.e. the geodesic equations.

    Before passing to the solution, let us first of all simplify the concept of metrics with `nested' structure:

    \begin{lem}\label{lem1}
        Metrics with `nested' structure can be represented by diagrams of the following types: in each bracket there are either strictly two other brackets or none at all. In other words, the  following two situations are allowed:
        \begin{equation}
        \begin{overpic}[scale=1.2,unit=0.5mm]{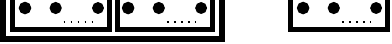}
        \end{overpic}
        \end{equation}
        
       \noindent(By the same rules as above, additional brackets can also be nested within the inner brackets in the diagram on the left\footnote{
        Here is an example of such metric on the flag manifold $\mathcal{F}_{1,2,\ldots,12}$:
        \begin{figure}[H]
        \centering
        \begin{overpic}[scale=1,unit=0.5mm]{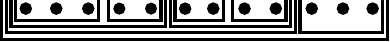}
        \end{overpic}
        \end{figure}
        }.)

    \end{lem}

    \begin{proof} 
        An arbitrary metric with `nested' structure can contain structures of the form
        \begin{equation}
            \begin{overpic}[scale=1.2,unit=0.5mm]{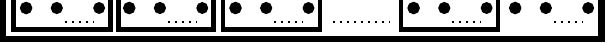}
            \end{overpic}
        \end{equation}
        \noindent The structure of the inner brackets is not specified.  

        In order to convert this diagram to a diagram of the type described in the statement of the Lemma, we add `fictitious' brackets whose coefficients in the final formulas will be set to zero:
        \begin{equation}
            \begin{overpic}[scale=1.2,unit=0.5mm]{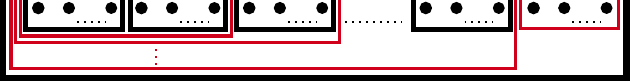}
            \end{overpic}
        \end{equation}
        This completes the proof.
    \end{proof}

  In solving the geodesic problem we will be using the bracket structure described in Lemma~\ref{lem1}. The action describing geodesics in $\mathcal{F}_{1,2,\ldots,N}$ is constructed from the metric by the formal substitution $d u_i\rightarrow \dot{u}_i$, in analogy with formula (\ref{FinalAction}). This action inherits the metric structure. From the equations of motion of the vectors $u_i$ one can easily derive the equations for the evolution of the scalar products $\dot{\Bar{u}}_i\circ u_j$ (see the example in Appendix~\ref{su3flagEOM}).
  Knowing these, one can determine the evolution of the vectors $\Bar{u}_i$ themselves. Indeed, using the `partition of unity' (which is a generalization of~(\ref{partunit}) for $N$ vectors), we have the following identity: $\Dot{\Bar{u}}_i=(\dot{\Bar{u}}_i\circ u_j) \Bar{u}_j$, which can be seen as a system of equations on the vectors $\Bar{u}_i$ ($i=1,\ldots,N$). 
  The main result is as follows:
    \begin{prop}\label{geodsolproof}
        One can construct explicitly the general solution to the geodesic equations for metrics on $\mathcal{F}_{1,2,\ldots,N}$ with `nested' structure.
    \end{prop}

    \begin{proof}

First, let us write down the equations of motion for the most general metric~(\ref{genmetr}). To do this, we introduce two matrices
\begin{gather}\label{omegaLdef}
\pmb{\omega}_{ij}:=\dot{\Bar{u}}_j\circ u_i\,,\quad\quad \pmb{L}_{ij}=\begin{cases}{1\over \alpha_{ij}}    \dot{\Bar{u}}_j\circ u_i\,,\quad i\neq j\\
\quad \quad 0\,,\quad\quad\;\;  i=j
\end{cases}
\end{gather}
For brevity, let us write $\pmb{L}={1\over \alpha}\left(\pmb{\omega}\right)$. It then follows from the equations of motion\footnote{See Appendix~\ref{su3flagEOM} for the explicit derivation in the $\mathrm{SU}(3)$ case.} that 
\begin{gather}\label{LOmegaEq}
{d \pmb{L}\over dt}=[\pmb{L}, \pmb{\omega}]. 
\end{gather}
These are generalized Euler equations\footnote{The conditions for Liouville integrability of such equations have been studied, e.g. in~\cite{MiFom}.} (in mechanics, $\pmb{\omega}$ and $\pmb{L}$ are the angular velocity and angular momentum of a solid body, respectively).

Let us now show that, for metrics with `nested' structure, these equations can be solved by induction. According to Lemma~\ref{lem1}, the diagram we are interested in has the form
    \begin{equation}
        \begin{overpic}[scale=1.2,unit=0.5mm]{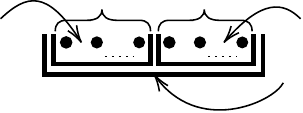}
            \put(-3,31){$1$}
            \put(123,32){$2$}
            \put(115,13){$3$}
            \put(38,44){$N$}
            \put(80,44){$K$}
        \end{overpic}
    \end{equation}
        \noindent 
        where brackets $1$ and $2$ can contain sub-brackets according to the rules above. In this case $\pmb{\omega}$ and $\pmb{L}$ may be split  as follows:
\begin{gather}
    \pmb{\omega}=\begin{pmatrix}
                \omega_1 & \omega\\
                -\omega^\dagger & \omega_2
            \end{pmatrix}\,,\quad\quad  \pmb{L}=\begin{pmatrix}
                L_1 & {1\over \xi}\omega\\
                -{1\over \xi}\omega^\dagger & L_2
            \end{pmatrix}\,,\\
            \textrm{where}\quad\quad L_1={1\over \alpha}(\omega_1)\quad \textrm{and}\quad L_2={1\over \beta}(\omega_2).
\end{gather}
Note that here $1\over \xi$ stands simply for multiplication by the coefficient of bracket $3$. The equations of motion~(\ref{LOmegaEq}) are accordingly split as follows:
\begin{gather}\label{omega12eq}
    {dL_1\over dt}
    =\left[ L_1 ,
\;\omega_1\right]\,,\quad\quad {dL_2\over dt}
=\left[ L_2 ,
\;\omega_2\right] \,,\\ \label{omegaeq}
    {1\over \xi} {d\omega\over dt}=\left(
    L_1-{1\over \xi}\omega_1\right)\,\omega-\omega\,\left(
    L_2-{1\over \xi}\omega_2\right).
\end{gather}
Thus, equations for $\omega_1$ and $\omega_2$ have decoupled. Moreover, the last equation is \emph{linear} in~$\omega$. Solutions with $\omega=\omega_1=0$ or $\omega=\omega_2=0$ correspond to the motion in the fibers of two natural forgetful bundles
\begin{center}
\begin{tabular}{ r c l }
 & $\mathcal{F}_{1,2,\ldots,N+K}$ &  \\ 
 \quad $\swarrow$  &   & $\searrow$ \\  
 $\mathcal{F}_{1,2, \ldots, K, N+K}$ &  & $\mathcal{F}_{K, \ldots, N+K}$    
\end{tabular}
\end{center}
It thus follows that the fibers are totally geodesic submanifolds.

Let us continue with the proof by induction. Suppose the equations for $\omega_1$ and $\omega_2$ have been solved. We also assume that the following representations have been constructed:
\begin{gather}\label{omega12pres}
\omega_1=-\dot{g}_1 g_1^{-1}\,,\quad\quad \omega_2=-\dot{g}_2 g_2^{-1}\,,
\end{gather}
where $g_{1,2}$ are explicitly known unitary matrices. In the first step of the induction we assume that there are no additional brackets inside brackets $1$ and $2$. Therefore, the operators $1\over\alpha$ and $1\over\beta$ are  multiplications by scalars. It then follows from equation~(\ref{omega12eq}) that $\omega_1=\mathrm{const}$ and $\omega_2=\mathrm{const}$, so that representation~(\ref{omega12pres}) holds, with $g_1 = e^{-\omega_1t}, g_2 = e^{-\omega_2t}$.

Next, we show that, under the inductive assumption, the equation for $\omega$ can be solved explicitly, and the entire matrix $\pmb{\omega}$ can be expressed in a form similar to~(\ref{omega12pres}). To do so, we substitute $\omega=g_1\hat{\omega}g_2^{-1}$. Equation~(\ref{omegaeq}) is transformed into
\begin{gather}
    {1\over \xi} {d\hat{\omega}\over dt}=A \,\hat{\omega}-\hat{\omega}\,B\,,\\
    \textrm{where}\quad\quad A=g_1^{-1} L_1
    g_1\quad \textrm{and}\quad B=g_2^{-1} L_2
    g_2.
\end{gather}
However, it follows from equations~(\ref{omega12eq}) and~(\ref{omega12pres}) that ${dA \over dt} = {dB \over dt} = 0$, i.e. $A$ and $B$ are constant skew-Hermitian matrices. Therefore the solution is
\begin{gather}\label{omegahatsol}
    \hat{\omega}=e^{\xi A t} \hat{\omega}_0e^{-\xi B t}.
\end{gather}
Simple transformations can be used to bring~$\pmb{\omega}$ to the desired form:
\begin{gather}
    \pmb{\omega}=\begin{pmatrix}
                -\dot{g}_1 g_1^{-1} & g_1 \hat{\omega}g_2^{-1}\\
                -\left(g_1 \hat{\omega}g_2^{-1}\right)^\dagger & -\dot{g}_2 g_2^{-1}
            \end{pmatrix}=-\dot{G} G^{-1}\,,\\ \nonumber \textrm{where} \quad\quad
             G=\begin{pmatrix}
                g_1 & 0\\
                0 & g_2
    \end{pmatrix}\,\begin{pmatrix}
                e^{\xi A t} & 0\\
                0 & e^{\xi B t}
            \end{pmatrix}\mathrm{exp}\left[-\begin{pmatrix}
                \xi\,A &  \hat{\omega}_0\\
               - \hat{\omega}_0^\dagger & \xi\,B
            \end{pmatrix}\,t\right]. 
\end{gather}
By the definition~(\ref{omegaLdef}), it follows that the matrix $\pmb{u}=(u_1\cdots u_{N+K})^t$ composed of the vectors $u_i$ satisfies the equation
$   \dot{\pmb{u}}=-\pmb{\omega}\pmb{u}=\dot{G}G^{-1}\pmb{u}
$. 
Hence
\begin{gather}\label{genSolution}
    \pmb{u}(t)=G(t) \pmb{u}_0\,.
\end{gather}
We have now completed the induction step and, in particular, found the evolution of the vectors $u_i$. This proves the proposition. 
\end{proof}

    Let us emphasize that the algorithm for constructing all solutions follows directly from the proof. As an example of the considered technique, we describe geodesics on the flag manifold $\mathcal{F}_{1,2,3,4}$ in Appendix~\ref{flag1234app}.
    %
    %
    %


    \section{Partial flag manifolds}\label{partialFlags}


    In this section we adapt the above construction to the case of partial flag manifolds. The core idea follows the concept introduced in Section \ref{reductionToCP2}: it suffices to consider a particular limit for the problem on the complete flag manifold $\mathcal{F}_{1,2,\ldots,N}$.

    Consider an $\SU(N)$ spin chain with $N$ spins. Suppose  part of the diagram of its Hamiltonian $\mathcal{H}$ is of the form
    \begin{equation}
        \begin{overpic}[scale=1.4,unit=0.5mm]{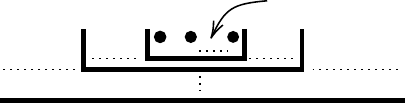}
           \put(74,37){$I$}
           \put(81,37){$I+1$}
           \put(107,37){$I+J$}
           \put(129,48){$1$}
        \end{overpic}
    \end{equation}
   In the Hamiltonian, bracket $1$ corresponds to the operator $\alpha_{I,I+J}H_{I,I+J}$ acting in the space $\mathrm{Sym}(p,N)^{\otimes (J+1)}$.

    Consider the limit $\alpha_{I,I+J} \rightarrow \infty$. By the algorithm described in Section \ref{SpinChain}, it follows that $\mathcal{H}$ has finite eigenvalues  for the unique irreducible representation in $\mathrm{Sym}(p,N)^{\otimes (J+1)}$, corresponding to the Young diagram
    \begin{equation}\label{JpRecDiag}
        J+1 \quad
        \underbrace{\!\!\!\!\!
        \begin{cases}
            \begin{ytableau}
            ~ & ~ & \dots & ~ \\
            ~ & ~ & \dots & ~ \\
            \vdots & \vdots& \vdots & \vdots\\
            ~ & ~ & \dots & ~ \\
            \end{ytableau}
        \end{cases} \!\!\!\!\!\!}_{p}:=\mathrm{Rec}(J+1,p,N).
    \end{equation}
    As a result, in this limit the Hilbert space is reduced to
    \begin{gather}\label{Hilbspacered}
         \mathrm{V}(p)\rightarrow \mathrm{Sym}(p,N)^{\otimes (I-1)}\otimes \mathrm{Rec}(J+1,p,N) \otimes \mathrm{Sym}(p,N)^{\otimes (N-I-J)}\,,
    \end{gather}
   since other states correspond to infinite eigenvalues. Note that when one multiplies several representations $\mathrm{Sym}(p,N)$ in $\mathrm{Sym}(p,N)^{\otimes (J+1)}$, the diagram~(\ref{JpRecDiag}) occurs uniquely when each $\mathrm{Sym}(p,N)$ factor is attached as a new row of length $p$. Therefore, if bracket~$1$ contains nested brackets within it, corresponding terms in the Hamiltonian automatically vanish on these representations due to the subtraction defined in (\ref{eigenvalueSU(N)}). This means that, if the coefficient of any bracket tends to infinity, the inner structure of this bracket becomes irrelevant.

    Based on the relation between metrics and Hamiltonians described in \ref{flagMetr}, all terms $|\Bar{u}_i\circ du_j|^2,$ $i,j=I,I+1,\ldots,I+J$ in the metric are suppressed\footnote{So that this metric is degenerate, if viewed as a metric on the complete flag manifold.} as $\alpha_{I,I+J} \rightarrow \infty$. Thus, the set of vectors $\{u_i\}_{i=I}^{I+J}$ enters the metric via the projector $\sum\limits_{i=I}^{I+J} u_i \otimes \Bar{u}_i$ onto the subspace spanned by these vectors, i.e. in an explicitly $\mathrm{U}(J+1)$-invariant manner.
    Thus, in the limit $p \rightarrow \infty$, the vectors $\{u_i\}_{i=I}^{I+J}$ corresponding to bracket~$1$ are defined up to a $\mathrm{U}(J+1)$-transformation.
 In other words, just as in the case of $\mathcal{F}_{1,2,3}$, the gauge group is enlarged from $\mathrm{U}(1)^{J+1}$ to $U(J+1)$, and the full set of vectors $\{u_a\}_{a=1}^N$ now parametrizes the partial flag manifold $\mathcal{F}_{1,\ldots,I-1,I+J,\ldots,N}$.
    
The procedure described above can be continued to obtain an arbitrary partial flag manifold. In this case the reduction of Hilbert space of the type~(\ref{Hilbspacered}) can be interpreted in terms of Lagrangian submanifolds, similar to what was done in Section~\ref{MetricsFlag}. First, we note that the rectangular diagram~(\ref{JpRecDiag}) corresponds to coherent states parameterized by points in the Grassmannian $\mathrm{Gr}(J+1, N)$ (see, e.g.~\cite{Bykov_2013} and~\cite{KirillovReview}). Therefore this time, when the coefficients of several terms in the Hamiltonian tend to infinity, the phase space of the spin chain is a product of Grassmannians instead of $(\mathbb{CP}^{N-1})^{N}$. Consequently, there is a natural generalization of Proposition~\ref{lagrtheorem}:
\begin{prop}[\cite{Bykov_2013,Affleck_2022}]
There exists an embedding \begin{gather}
    \mathcal{F}_{d_1,d_2,\ldots,d_m=N} \hookrightarrow \prod\limits_{i=1}^m \mathrm{Gr}(n_i,N)\,,\quad \textrm{where}\quad n_i=d_i-d_{i-1}, \,d_0=0\,.
\end{gather}
Moreover, this embedding is Lagrangian w.r.t. a natural symplectic structure on the product of Grassmannians.  
\end{prop}

The embedding is given by the condition that planes corresponding to the Grassmannians in the r.h.s. are pairwise orthogonal. The rest of the proof proceeds parallel to the proof of Proposition~\ref{lagrtheorem}.

The metric with `nested' structure on a manifold of partial flags is obtained from the metric with nested structure on the complete flag manifold $\mathcal{F}_{1,2,\ldots,N}$ by taking the limit described above. This leads us to the main result of this section:
    \begin{prop}
        Suppose a manifold of partial flags is equipped with a metric with `nested' structure. In this case the general solution to the corresponding geodesic equations and the spectrum of the associated Laplace-Beltrami operator can be found explicitly.
    \end{prop}

    \begin{proof}
        This follows directly from the considerations above and from Propositions \ref{LaplaceBeltramiSpect} and~\ref{geodsolproof} for complete flags.

        Let us make a brief remark. As $\alpha_{I,I+J} \rightarrow \infty$, all contributions from $\{\dot{\Bar{u}}_i\circ u_j\}_{i,j=I}^{I+J}$ in the solution (\ref{genSolution}) are suppressed, except for one that corresponds to the term $e^{-\Tilde{\omega}t}$ with the constant matrix $\Tilde{\omega}=||\dot{\Bar{u}}_i\circ u_j||_{i,j=I}^{I+J}=\mathrm{const}$. This term  acts as a rotation of the vectors $\{u_i\}_{i=I}^{I+J}$, which is a pure gauge transformation and therefore can be excluded, just as in the case of the transition from $\mathcal{F}_{1, 2, 3}$ to $\mathbb{CP}^2$ discussed earlier in Section~\ref{reductionToCP2}.
        \end{proof}

\section{Conclusion and outlook}

In this paper, two related problems are considered, namely, the description of geodesics and the calculation of the spectrum of the Laplace-Beltrami operator on flag manifolds. It is shown that for metrics with `nested' structure the solutions to both problems can be found explicitly. As an intermediate step, we also construct auxiliary $\SU(N)$ spin chains whose Hamiltonians provide finite-dimensional approximations to the Laplace-Beltrami operator.

The above two problems are closely linked. Essentially, the Laplace-Beltrami operator can be seen as a quantum Hamiltonian corresponding to the classical geodesic problem. However, fully understanding this connection is not an easy task and requires the use of quasi-classical techniques (see, for example, ~\cite{Gutzwiller, Camporesi}).

It would as well be interesting to generalize our results to other systems. Specifically, although we have studied a wide range of metrics on flag manifolds, these metrics do not cover the entire space of invariant metrics in general. The obvious question is what happens for other metrics and whether explicit solutions exist in those cases. Furthermore, we have only considered the $\mathrm{SU}(N)$-case, but it is likely that similar methods could be applied to flag manifolds of $\mathrm{O}(n)$ and $\mathrm{Sp}(n)$~\cite{AlekseevskyGO}.

One of the main motivations for investigating the integrability of geodesic equations in this paper is in the possible two-dimensional generalization to the case of sigma models. The most well-studied case is that of symmetric target spaces~\cite{ZakharovMikhailov, Uhlenbeck, HitchinHarmonic} (see the reviews~\cite{Helein, Guest}), where the metric is unique up to scaling, and geodesics correspond to the orbits of one-dimensional subgroups of the isometry group. As we have seen above, in the case of non-symmetric target spaces, such as flag manifolds, there are metrics for which the explicit solution to the geodesic equations can still be found in a  relatively simple way. In this context, a curious question arises: are there similar integrable two-dimensional sigma models with metrics different from the normal one (where integrability seems to occur~\cite{BykovNonsymm, Affleck_2022})? For example, integrability has been proven for the case of odd-dimensional spheres with an arbitrary size of the Hopf fiber~\cite{BassoRej}. On the other hand, it is well-known that metrics for which geodesic equations are integrable do not always lead to integrable sigma models (see, for example,~\cite{StepanchukTseytlin, ChervonyiLunin}). We plan to return to these questions in the future.
    

\vspace{0.5cm}\noindent
    \textbf{Acknowledgments.} Sections 1-2 were written with the support of the Foundation for the Advancement of Theoretical Physics and Mathematics «BASIS». Sections 3–6 were supported by the Russian Science Foundation grant \href{https://rscf.ru/project/22-72-10122/}{№ 22-72-10122}. We would like to thank S. Derkachov, A.~Goncharov, S.~Gorchinskiy, V. Krivorol, M. Markov, A.~Selemenchuk for useful discussions.

\vspace{1cm}

    \appendix

    \section{Equations of motion}\label{su3flagEOM}


    The equations of motion for the theory with Lagrangian (\ref{alphabetaLagr}) are
    \begin{eqnarray}
     && u_1:\quad  \frac{1}{\alpha}\dot{\Bar{u}}_2 (\Bar{u}_1\circ \dot{u}_2) + \Bar{u}_i\lambda^{i1}=0,
        \nonumber
        \\
     &&  u_2:\quad  \frac{1}{\alpha}
        \left(
            \Bar{u}_1 (\Ddot{\Bar{u}}_2 \circ u_1) +
            \Bar{u}_1 (\dot{\Bar{u}}_2 \circ \dot{u}_1) +
            \dot{\Bar{u}}_1 (\dot{\Bar{u}}_2 \circ u_1)
        \right)
        -
        \Bar{u}_i\lambda^{i2}=0,
        \\
     &&   u_3:\quad
        \frac{1}{\beta}
        \left( 
            \Ddot{\Bar{u}}_3 -
            \Bar{u}_3 (\Ddot{\Bar{u}}_3 \circ u_3) -
            \Bar{u}_3 (\dot{\Bar{u}}_3 \circ \dot{u}_3) - 
            \dot{\Bar{u}}_3 (\dot{\Bar{u}}_3 \circ u_3) + 
            \dot{\Bar{u}}_3 (\Bar{u}_3 \circ \dot{u}_3)
        \right)
        -
        \Bar{u}_i\lambda^{i3}=0.
        \nonumber 
    \end{eqnarray}
Here we have omitted the conjugate equations. 
    By scalar multiplying the above equations by the vectors $u_1, u_2, u_3$, we obtain evolution equations for the scalar products~$A_{ij}=\Bar{u}_i \circ \Dot{u}_j$:
    \begin{eqnarray}
        &&\!\!\!\!\!\!\!\!\!\!\!\!\!\!\!\frac{d}{dt}
        \left(
            \Bar{u}_2 \circ \Dot{u}_1
        \right)
        -
        \Bar{u}_2 \circ \dot{u}_1 ( \Bar{u}_1 \circ \dot{u}_1 - \Bar{u}_2 \circ \dot{u}_2 )
        =0,
        \nonumber
        \\
        &&\!\!\!\!\!\!\!\!\!\!\!\!\!\!\!\frac{1}{\beta}
        \frac{d}{dt}
        \left(
            \Bar{u}_3 \circ \dot{u}_1
        \right)
        +
        \frac{1}{\beta}
        \dot{\Bar{u}}_3 \circ \dot{u}_1
        + 
        \frac{2}{\beta} 
        (\Bar{u}_3 \circ \dot{u}_3) (\Bar{u}_3 \circ \dot{u}_1) 
        + 
        \frac{1}{\alpha} 
        (\Bar{u}_3 \circ \dot{u}_2) (\Bar{u}_2 \circ \dot{u}_1)
        =0, 
        \label{EoM}
        \\
        &&\!\!\!\!\!\!\!\!\!\!\!\!\!\!\!\frac{1}{\beta}
        \frac{d}{dt}
        \left(
            \Bar{u}_3 \circ \dot{u}_2
        \right)
        +
        \frac{1}{\beta}
        \dot{\Bar{u}}_3 \circ \dot{u}_2
        + 
        \frac{2}{\beta} 
        (\Bar{u}_3 \circ \dot{u}_3) (\Bar{u}_3 \circ \dot{u}_2) 
        +
        \frac{1}{\alpha} 
        (\Bar{u}_3 \circ \dot{u}_1) (\Bar{u}_1 \circ \dot{u}_2)
        =0.
        \nonumber 
    \end{eqnarray}
    It is useful to note that, given the gauge fixing conditions $\Bar{u}_i \circ  \Dot{u}_i = 0$, one can write
    \begin{eqnarray}
&&\dot{\Bar{u}}_3 \circ \dot{u}_1 = \\ \nonumber
&& =
        - (\Bar{u}_3 \circ \dot{u}_1) (\Bar{u}_1 \circ \dot{u}_1) 
        - (\Bar{u}_3 \circ \dot{u}_2) (\Bar{u}_2 \circ \dot{u}_1) 
        - (\Bar{u}_3 \circ \dot{u}_3) (\Bar{u}_3 \circ \dot{u}_1)
        =- (\Bar{u}_3 \circ \dot{u}_2) (\Bar{u}_2 \circ \dot{u}_1),\\
        &&\dot{\Bar{u}}_3 \circ \dot{u}_2 = - (\Bar{u}_3 \circ \dot{u}_1) (\Bar{u}_1 \circ \dot{u}_2).\nonumber
    \end{eqnarray}
    Thus, the system (\ref{EoM}) takes the form
    \begin{eqnarray}
        &&\frac{d}{dt}
        \left(
            \Bar{u}_2 \circ \Dot{u}_1
        \right)
        =0,
        \nonumber
        \\
        &&\frac{1}{\beta}
        \frac{d}{dt}
        \left(
            \Bar{u}_3 \circ \dot{u}_1
        \right)
        +
        \left(\frac{1}{\alpha} - \frac{1}{\beta}\right)
        (\Bar{u}_3 \circ \dot{u}_2) (\Bar{u}_2 \circ \dot{u}_1)
        =0, 
        \\
        &&\frac{1}{\beta}
        \frac{d}{dt}
        \left(
            \Bar{u}_3 \circ \dot{u}_2
        \right)
        +
        \left(\frac{1}{\alpha} - \frac{1}{\beta}\right)
        (\Bar{u}_3 \circ \dot{u}_1) (\Bar{u}_1 \circ \dot{u}_2)
        =0,
        \nonumber 
    \end{eqnarray}
    which, given the notation, is equivalent to (\ref{EvolutionEquation}).

  \section{Ground state}\label{groundState}

  In this Appendix we will identify the ground state (i.e., the state of lowest energy) of the system described in Section \ref{SpinChain}. A natural assumption is that it corresponds to the singlet in $\mathrm{Sym}(p,N)^{\otimes N}$, whose Young diagram consists of $N$ rows of length~$p$. To this end, according to the algorithm for computing the eigenvalues of $\mathcal{H}$ (see Section~\ref{SpinChain}), if the coefficients of all brackets are positive\footnote{This strengthens the requirements on the coefficients as compared to those introduced in Section~\ref{SpinChain}. In the general case one can use an alternative proof based on the connection with the Laplace-Beltrami operator and the non-negative definiteness of the latter.}, it is sufficient to show that
  \begin{gather}\label{groundStateProp}
      C_2(p_1,\ldots,p_a,0,\ldots,0)-C_2(p'_1=p,\ldots,p'_a=p,0,\ldots,0) \geq 0
  \end{gather}
  for all sets $p_1\geq p_2\geq \ldots \geq p_a\geq 0$, such that $\sum\limits_{i=1}^a p_i = ap$, and the inequality is saturated if and only if $p_i=p$ for all $i$.

  By definition (\ref{quadraticCasimirSU(N)}), it follows that equation (\ref{groundStateProp}) can be rewritten as
  \begin{gather}\label{minimEq}
      \sum_{i=1}^a\left((\delta p_i-i)^2-i^2\right) \geq 0,
  \end{gather}
  where $\delta p_i = p_i - p$. Let us find the minimum of the l.h.s. of (\ref{minimEq}) subject to the following constraints on $\delta p_j$:
  \begin{gather}
      \sum\limits_{i=1}^a \delta p_i = 0,\label{firstCondition}\\
      \delta p_1\geq  \ldots \geq \delta p_a\geq -p.\label{secondCondition}
  \end{gather}
  It is convenient to start by considering the first condition and then move on to the second one.

  The minimum of (\ref{minimEq}) under condition (\ref{firstCondition}) is reached when 
  \begin{gather}\label{deltap0}
  \delta p_i = \delta p^0_i := i - {a+1 \over 2}\,,\quad\quad i=1,\, \ldots, a\,.
  \end{gather}
  In this case (\ref{secondCondition}) is not satisfied, though. However,  the expression that we seek to minimize can now be reduced to
  \begin{gather}\label{lastMinimEq}
      \sum_{i=1}^a\left((\delta p_i-\delta p^0_i)^2+(\delta p^0_i-i)^2-i^2\right).
  \end{gather}
  This takes into account that  $\sum\limits_{i=1}^a (\delta p_i-\delta p^0_i)(\delta p^0_i-i)=-\frac{a+1}{2}\sum\limits_{i=1}^a (\delta p_i-\delta p^0_i)=0$ (see~(\ref{deltap0})).  
  Dropping the constants and using the notation $j(i):=a+1-i$, one may rewrite~(\ref{lastMinimEq}) as
  \begin{gather}
      \sum_{i=1}^{\left[a/2\right]}\left((\delta p_i-\delta p^0_i)^2+(\delta p_{j(i)}-\delta p^0_{j(i)})^2 \right)+\delta p^2_{(a+1)/2},
  \end{gather}
  where $\left[a/2\right]$ is the integer part of  $a/2$ (the last term is present only if $(a+1)/2$ is an integer). Removing the brackets and using $\delta p^0_i = - \delta p^0_{j(i)}$, we get
  \begin{gather}\label{lastMinimEq2}
        \sum_{i=1}^{\left[a/2\right]}\left(\delta p_i^2+\delta p_{j(i)}^2 + \left(\delta p^0_i\right)^2+ \left(\delta p^0_{j(i)}\right)^2 + 2 \delta p^0_{j(i)} \left(\delta p_i - \delta p_{j(i)}\right) \right)+\delta p^2_{(a+1)/2},
  \end{gather}
  The term $2 \delta p^0_{j(i)} \left(\delta p_i - \delta p_{j(i)}\right) \geq 0$, since $\delta p^0_{j(i)} > 0$ and $\delta p_i \geq \delta p_{j(i)}$ for $i = 1,\ldots, \left[a/2\right]$ (see (\ref{secondCondition}), (\ref{deltap0}) and the definition of $j(i)$). Thus,  (\ref{lastMinimEq2}) reaches its minimum at $\delta p_i=0$. In other words, the minimum value of the l.h.s. of (\ref{groundStateProp}) is zero, and is reached when $p_i=p$ for all $i$, which completes the proof.

  \section{Flag manifold $\mathcal{F}_{1,2,3,4}$}\label{flag1234app}
    %
    %
    %
    Let us use the flag manifold $\mathcal{F}_{1,2,3,4}$ to illustrate the method for finding geodesics described in Section~\ref{geodesicsec} (see the proof of Proposition~\ref{geodsolproof}). Consider the metric corresponding to the  diagram
    \begin{equation}
        \begin{overpic}[scale=1.2,unit=0.5mm]{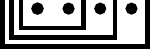}
        \end{overpic}
    \end{equation}
    This metric is of the form
    \begin{gather}
        ds^2=
        \frac{1}{\alpha}\left|\Bar{u}_1\circ du_2\right|^2+
        \frac{1}{\beta}\left(
        \left|\Bar{u}_2\circ du_3\right|^2+
        \left|\Bar{u}_1\circ du_3\right|^2
        \right)
        +\\ \nonumber  +\;\;
        \frac{1}{\gamma}
        \left(
        \left|\Bar{u}_3\circ du_4\right|^2+
        \left|\Bar{u}_2\circ du_4\right|^2+
        \left|\Bar{u}_1\circ du_4\right|^2
        \right)\,,
    \end{gather}
    where we have redefined the coefficients of the brackets for convenience.

    Let $A_{ij} := \Bar{u}_i \circ \dot{u}_j$. Then, $A_{ji}=-\Bar{A}_{ij}$. The evolution equations for $A_{ij}$ are given by
    \begin{gather}
        \frac{d}{dt}A_{21} = 0,\label{flag14FirstEq}\\
        \frac{d}{dt}
        \begin{pmatrix}
            A_{31}\\
            A_{32}
        \end{pmatrix}
        =
        \begin{pmatrix}
            0 & \left(1-\frac{\beta}{\alpha}\right)A_{21}\\
            \left(\frac{\beta}{\alpha}-1\right)\Bar{A}_{21} & 0
        \end{pmatrix}
        \begin{pmatrix}
            A_{31}\\
            A_{32}
        \end{pmatrix},\\
        \frac{d}{dt}
        \begin{pmatrix}
            A_{41}\\
            A_{42}\\
            A_{43}
        \end{pmatrix}
        =
        \begin{pmatrix}
            0 & \left(1-\frac{\gamma}{\alpha}\right)A_{21} & \left(1-\frac{\gamma}{\beta}\right)A_{31}\\
            \left(\frac{\gamma}{\alpha}-1\right)\Bar{A}_{21} & 0 & \left(1-\frac{\gamma}{\beta}\right)A_{32}\\
            \left(\frac{\gamma}{\beta}-1\right)\Bar{A}_{31} & \left(\frac{\gamma}{\beta}-1\right)\Bar{A}_{32} & 0
        \end{pmatrix}
        \begin{pmatrix}
            A_{41}\\
            A_{42}\\
            A_{43}
        \end{pmatrix}.\label{flag14LastEq}
    \end{gather}
    
    Let $A_{ij}^0\equiv A_{ij} \big|_{t=0}$ and introduce the matrices
    \begin{gather}
        \mathcal{M}_1\equiv
        \exp \left[
        \begin{pmatrix}
            0 & \frac{1}{\alpha}A^0_{21}\\
            -\frac{1}{\alpha}\Bar{A}^0_{21} & 0
        \end{pmatrix}(\alpha - \beta)t
        \right],
        \\
        \mathcal{M}_2\equiv
        \exp
        \left[
        \begin{pmatrix}
            0 & \frac{1}{\alpha}A^0_{21} & \frac{1}{\beta}A^0_{31}\\
            -\frac{1}{\alpha}\Bar{A}^0_{21} & 0 & \frac{1}{\beta}A^0_{32}\\
            -\frac{1}{\beta}\Bar{A}^0_{31} & -\frac{1}{\beta}\Bar{A}^0_{32} & 0
        \end{pmatrix}(\beta-\gamma)t    
        \right],
        \\
        \mathcal{M}_3\equiv
        \exp 
        \left[
        \begin{pmatrix}
            0 & \frac{1}{\alpha}A^0_{21} & \frac{1}{\beta}A^0_{31} & \frac{1}{\gamma}A^0_{41}\\
            -\frac{1}{\alpha}\Bar{A}^0_{21} & 0 & \frac{1}{\beta}A^0_{32} & \frac{1}{\gamma}A^0_{42}\\
            -\frac{1}{\beta}\Bar{A}^0_{31} & -\frac{1}{\beta}\Bar{A}^0_{32} & 0 & \frac{1}{\gamma}A^0_{43}\\
            -\frac{1}{\gamma}\Bar{A}^0_{41} & -\frac{1}{\gamma}\Bar{A}^0_{42} & -\frac{1}{\gamma}\Bar{A}^0_{43} & 0
        \end{pmatrix}(\gamma-\delta)t 
        \right].
    \end{gather}

    Then the solutions to the equations (\ref{flag14FirstEq})-(\ref{flag14LastEq}) are
    \begin{gather}
        A_{21} = A^0_{21},\\
        \begin{pmatrix}
            A_{31}\\
            A_{32}
        \end{pmatrix} = \mathcal{M}_1
        \begin{pmatrix}
            A^0_{31}\\
            A^0_{32}
        \end{pmatrix},
        \\
        \begin{pmatrix}
            A_{41}\\
            A_{42}\\
            A_{43}
        \end{pmatrix}
        =
        \begin{pmatrix}
        \mathcal{M}_1 &
        \\
        & 1
        \end{pmatrix}
        \mathcal{M}_2
        \begin{pmatrix}
            A^0_{41}\\
            A^0_{42}\\
            A^0_{43}
        \end{pmatrix}.
    \end{gather}    
    
    Finally, the system of equations for the evolution of the set of vectors $\{u_i\}_{i=1}^4$ 
    \begin{gather}
        \frac{d}{dt}
        u_i = A_{ij}u_j
    \end{gather}
    can be solved explicitly:
    \begin{gather}
        \begin{pmatrix}
            u_1\\
            u_2\\
            u_3\\
            u_4
        \end{pmatrix}(t)= 
        \begin{pmatrix}
        \mathcal{M}_1 & & 
        \\
        & 1 & \\
        & & 1 
        \end{pmatrix}
        \begin{pmatrix}
        \mathcal{M}_2 &\\
        & 1 
        \end{pmatrix}
        \mathcal{M}_3
        \begin{pmatrix}
            u_1\\
            u_2\\
            u_3\\
            u_4
        \end{pmatrix}(0).
    \end{gather}

    In particular, it can be seen that smaller complete flag manifolds (fibers of forgetful bundles) are totally geodesic submanifolds.

\vspace{1cm}    
    \setstretch{0.8}
    \setlength\bibitemsep{5pt}
    \printbibliography

\end{document}